\documentclass{article}

\usepackage{amsmath}
\usepackage{amsfonts}
\usepackage{amssymb}
\usepackage{amsthm}
\usepackage{stmaryrd}
\usepackage{centernot}
\usepackage{tikz}
\usepackage{hyperref}

\newcommand{\mlabel}[1]{\label{#1}}

\newcommand{\myskip}{\bigskip\noindent}

\newcommand{\eqclass}[1]{[{#1}]} %

\newcommand{\vald}{\mathsf{D}} %
\newcommand{\valdi}{\mathsf{D}'}
\newcommand{\valdii}{\mathsf{D}''}
\newcommand{\valfx}{\mathsf{F}} %
\newcommand{\valf}{\mathsf{F}'} %
\newcommand{\empt}{\ominus} %
\newcommand{\valv}{v} %

\newcommand{\setv}{\mathbb{V}} %
\newcommand{\setn}{\mathbb{N}} %
\newcommand{\setd}{[\vald]} %

\newcommand{\fsreplacement}[3]{{#1}_{[{#2}\rightarrow{#3}]}} %

\newcommand{\parent}{\mathord{\shortuparrow}} %
\newcommand{\parentf}[1]{\parent({#1})} %
\newcommand{\fsbroken}{\bot} %
\newcommand{\FS}{\Phi} %
\newcommand{\aFS}{\,\FS} %
\newcommand{\nn}{n'} %
\newcommand{\pn}{\parentf{n}} %

\newcommand{\cbrk}{\epsilon_\varnothing} %

\newcommand{\caaa}[3]{\langle{[#1],#2,#3}\rangle}
\newcommand{\caaaa}[3]{\langle{[#1],#2,#3}\rangle}

\newcommand{\cbba}[1]{\caaa{\empt}{\empt}{#1}}
\newcommand{\cbfa}[1]{\caaa{\empt}{\valfx}{#1}}
\newcommand{\cbda}[1]{\caaa{\empt}{\vald}{#1}}
\newcommand{\cfba}[1]{\caaa{\valfx}{\empt}{#1}}

\newcommand{\cfda}[1]{\caaa{\valfx}{\vald}{#1}}
\newcommand{\cdba}[1]{\caaa{\vald}{\empt}{#1}}
\newcommand{\cdfa}[1]{\caaa{\vald}{\valfx}{#1}}
\newcommand{\cdda}[1]{\caaa{\vald}{\vald}{#1}}

\newcommand{\cxyaa}[1]{\caaa{x}{y}{#1}}
\newcommand{\cxynv}{\caaa{x}{y}{n}}
\newcommand{\cxynnv}{\caaa{x}{y}{\nn}}
\newcommand{\cxypnv}{\caaa{x}{y}{\pn}}
\newcommand{\cxwnv}{\caaa{x}{w}{n}}
\newcommand{\czwnv}{\caaa{z}{w}{n}}
\newcommand{\czwnnv}{\caaa{z}{w}{\nn}}
\newcommand{\czwpnv}{\caaa{z}{w}{\pn}}
\newcommand{\czwmv}{\caaa{z}{w}{m}}

\newcommand{\cqrmv}{\caaa{q}{r}{m}}
\newcommand{\cqrov}{\caaa{q}{r}{o}}

\newcommand{\cc}{\,}

\newcommand{\descendant}{\prec} %
\newcommand{\descendantEq}{\preceq}
\newcommand{\ancestor}{\succ} %
\newcommand{\ancestorEq}{\succeq}

\newcommand{\eqext}{\sqsubseteq} %
\newcommand{\eqnrw}{\sqsupseteq} %
\newcommand{\strext}{\sqsubset} %
\newcommand{\strnrw}{\sqsupset} %
\newcommand{\nequiv}{\not\equiv}

\newcommand{\indep}{\mathrel{\wr\wr}} %
\newcommand{\unrel}{\indep} %
\newcommand{\nindep}{\mathrel{\centernot{\wr\wr}}} %
\newcommand{\nunrel}{\nindep} %

\newcommand{\worksmsign}{\vartriangleleft} %
\newcommand{\worksmeqsign}{\trianglelefteq} %
\newcommand{\worksm}[2]{{#1}\worksmeqsign{#2}} %

\newcommand{\worksmbx}[2]{\worksm{\{{#1}\}}{#2}}
\newcommand{\worksmbb}[2]{\worksm{\{{#1}\}}{\{{#2}\}}}

\newcommand{\wrks}[1]{\cbrk\strext{#1}} %
\newcommand{\wrksx}[1]{{#1}\strnrw\cbrk} %
\newcommand{\worksmnil}[1]{{\cbrk}\worksmsign{#1}} %
\newcommand{\worksmnilb}[1]{{\cbrk}\worksmsign{\{{#1}\}}} %

\newcommand{\emptyseq}{\lambda} %

\newcommand{\orderrel}{\ll}
\newcommand{\ordersetsign}{\vec{P}}
\newcommand{\orderset}[1]{\ordersetsign({#1})}
\newcommand{\ordsign}{\vec{c}} %
\newcommand{\ords}[1]{\ordsign\{{#1}\}} %
\newcommand{\ordp}[1]{\ordsign\,({#1})} %

\newcommand{\seqset}[1]{\mathcal{#1}} %
\newcommand{\sqs}[1]{\seqset{#1}}

\newcommand{\recchar}[2]{\mathcal{R}({#1}|{#2})}
\newcommand{\reca}{\recchar{A}{B}} %
\newcommand{\recb}{\recchar{B}{A}} %

\newcommand{\mynegsp}{\nobreak\hspace{-1pt plus 1pt}\nobreak}
\newcommand{\acbnp}{A\mynegsp\cap\mynegsp B} %
\newcommand{\acb}{\ordp{\acbnp}} %
\newcommand{\ambnp}{A\mynegsp\setminus\mynegsp B}
\newcommand{\amb}{\ordp{\ambnp}}
\newcommand{\bmanp}{B\mynegsp\setminus\mynegsp A}
\newcommand{\bma}{\ordp{\bmanp}}

\newcommand{\Dom}[1]{\textrm{Dom}({#1})}

\newcommand{\whr}{\mathrel{|}}

\newcommand{\rot}[3]{#3#2#1} %
\newcommand{\undersc}{\texttt{\detokenize{_}}}

\theoremstyle{plain}
\newtheorem{myth}{Theorem}

\theoremstyle{definition}
\newtheorem{mydefaux}[myth]{Definition}

\newenvironment{mydef}{\pushQED{\qed}\mydefaux}{\popQED\endmydefaux}

\theoremstyle{plain}
\newtheorem{mylem}[myth]{Lemma}

\theoremstyle{plain}
\newtheorem{mycor}[myth]{Corollary}

\theoremstyle{plain}
\newtheorem{myclm}[myth]{Claim}

\theoremstyle{plain}
\newtheorem{myaxproof}{Rule}

\newcommand{\axaxseparatecommute}{
Commands on incomparable nodes commute:
$\cxynv\cc\czwmv \equiv \czwmv\cc\cxynv$
when $n\unrel m$.
}

\newcommand{\axaxseparatenobreaks}{
Commands on incomparable nodes also do not break all filesystems:
$\wrksx{\cxynv\cc\czwmv}$
when $n\unrel m$.
}

\newcommand{\axaxsamebreaks}{
Commands on the same node break every filesystem if their types are incompatible:
$\cxynv\cc\czwnv \equiv \cbrk$
when $[y]\ne [z]$.
}

\newcommand{\axaxsameemptyseq}{
Commands on the same node are extended by an empty sequence if their types are compatible,
and their outer types represent an assertion command:
$\cxynv\cc\czwnv \eqext \emptyseq$ 
when $[y]=[z]$, and
$[x]$ and $[w]$ are either both $[\empt]$ or both $[\vald]$.
}

\newcommand{\axaxsamesinglec}{
Command pairs on the same node are equivalent to a single command if their types are compatible,
and their outer types do not represent an assertion command:
$\cxynv\cc \czwnv \equiv \cxwnv$
when $[y]=[z]$, and
$[x]$ and $[w]$ are neither both $[\empt]$ nor both $[\vald]$.
}

\newcommand{\axaxdirectchildbreaks}{
\begin{sloppypar}
Commands on a parent and a child node break every filesystem
if the pair is not a construction pair
and the commands do not simply assert a directory at the parent
or an empty node at the child:
$\cxypnv\cc\czwnv \equiv \cbrk$ when $\cxypnv\cc\czwnv$ is not a construction pair, 
$\cxypnv\neq\cdda{\pn}$ and $\czwnv\neq\cbba{n}$.
\end{sloppypar}
}

\newcommand{\axaxdirectparentbreaks}{
\begin{sloppypar}
Commands on a child and parent node break every filesystem
if the pair is not a destruction pair
and the commands do not simply assert an empty node at the child
or a directory at the parent:
$\cxynv\cc\czwpnv \equiv \cbrk$ when $\cxynv\cc\czwpnv$ is not a destruction pair,
$\cxynv\neq\cbba{n}$ and $\czwpnv\neq\cdda{\pn}$.
\end{sloppypar}
}

\newcommand{\axaxdistantrelbreaks}{
Commands on distant relatives break all filesystems
if the commands do not simply assert a directory at the ancestor
or an empty node at the descendant:
$\cxynnv\cc\czwnv \equiv \cbrk$
and $\czwnv\cc\cxynnv \equiv\cbrk$
when $\nn\descendant n$ and $\nn\neq\parent(n)$, and $\cxynnv\neq\cdda{\nn}$ and $\czwnv\neq\cbba{n}$.
}

\newcommand{\axaxchildassert}{
An assertion command can be freely added on a descendant node
next to a command that does not simply assert a directory:
\[ \cbba{n}\cc\cxynnv \equiv \cxynnv \equiv \cxynnv\cc\cbba{n} \] 
when $\nn\descendant n$ and $\cxynnv\neq\cdda{\nn}$.
}

\newcommand{\axaxparentassert}{
An assertion command can be freely added on an ancestor node
next to a command that does not simply assert an empty node:
\[ \cdda{\nn}\cc\cxynv \equiv \cxynv \equiv \cxynv\cc\cdda{\nn} \]
when $\nn\descendant n$ and $\cxynv\neq\cbba{n}$.
}

\newcommand{\axaxassert}{
All assertion commands can be removed as the empty sequence extends them:
$\cdda{n}\eqext\emptyseq$ and $\cbba{n}\eqext\emptyseq$.
}

\title{Algebraic File Synchronization: Adequacy and Completeness}
\author{Elod Pal Csirmaz\\
\texttt{\rot{\rot{maz.}{csir}{ep}com}{@}{elod}}}
\date{}

\begin{document}
\maketitle
\begin{abstract}

With distributed computing and mobile applications,
synchronizing diverging replicas of data structures is a more and more common problem.
We use algebraic methods to reason about filesystem operations, 
and introduce a simplified definition of conflicting updates to filesystems.
We also define algorithms for update detection and reconciliation
and present rigorous proofs that they not only work as intended,
but also cannot be improved on.

To achieve this, we introduce a novel, symmetric set of filesystem commands
with higher information content,
which removes edge cases
and increases the predictive powers of our algebraic model.
We also present a number of generally useful classes and properties
of sequences of commands.

While the results are often intuitive,
providing exact proofs for them is far from trivial.
They contribute to our understanding of this special type of algebraic model,
and toward building more complete algebras
of filesystem trees 
and extending algebraic approaches to other data storage protocols.
They also form a theoretical basis for
specifying
and guaranteeing the error-free operation
of applications that implement an algebraic approach to synchronization.

\myskip

\myskip
{\bf MSC classes:} %
08A70, %
68M07, %
68M14, %
68P05, %
68P20 %

\myskip
{\bf ACM classes:} %
D.4.3, %
E.5, %
F.2.2, %
G.2 %

\end{abstract}

\section{Introduction}

Synchronization of data structures
is a mechanism behind many services we take for granted today:
accessing and editing calendar events, documents and spreadsheets
on multiple devices, version control systems and
geographically distributed internet and web services that
guarantee a fast and responsive user experience.

In this paper we investigate a command-based approach
to synchronizing filesystem trees.
Presented with multiple copies (replicas) of the same filesystem that have been independently modified,
we regard the aim of the synchronizer to modify these replicas further
to make them as similar as possible.
The synchronization algorithm we describe
follows the two main steps described by Balasubramaniam and Pierce \cite{BP}:
update detection, where we identify modifications that have been applied to the replicas
since they diverged and represent them as sequences of commands;
and reconciliation, where we identify commands that can be propagated to other replicas and do so.
The remaining commands, which could not be propagated, represent conflicting updates;
resolving these requires a separate system or human intervention.

We extend significantly the results of the previous work done
with Prof. Norman Ramsey
on an algebraic approach to synchronization
\cite{NREC},
and add to the theoretical understanding
of synchronization by providing rigorous proofs that the algorithms
we describe work as intended
and offer a solution in our framework that cannot be improved on.

A central problem of command-based synchronizers 
during both update detection and reconciliation
is the ordering (scheduling) of commands.
If update detection is based on comparing the original state of the filesystem
to its new state, then we can easily collect a set of updates,
but we need to order these in a way that they could be applied
to the filesystem without causing an error 
(for example, a directory needs to be created before a file can be created under it).
Similarly, during reconciliation when we consider all the commands that
have been applied to the different replicas, we need to find a way of
merging these and creating a global ordering that we can apply
to the filesystems.
In fact, an insight in \cite{NREC} is that
if the ordering changes the effect of two commands,
then under certain circumstances they are not compatible and will give rise to conflicts.
Accordingly, in this paper command pairs, their commutativity and other properties
play a significant role.

In terms of the usual classification of synchronizers \cite{TSR, PV, SSH},
our approach to reconciliation is therefore operation-based as we investigate sets of commands,
while the update detection algorithm makes the full synchronizer 
similar to state-based synchronizers 
as its starting point is the states of the replicas themselves.
Moreover, our ordering algorithms only depend on properties of commands 
such as commutativity or idempotence,
and therefore they can be classified as semantic schedulers.
This is in opposition to syntactic scheduling, which uses 
the origin and the original temporal order of updates, 
and which may give rise to otherwise resolvable conflicts \cite{SSH}.
In general, ordering plays a central role in operation-based synchronizers.
We refer to an excellent survey by Saito and Shapiro of so-called optimistic replication algorithms \cite{SSH};
and, as examples,
to IceCube, where multiple orders are tested to find an acceptable one (\cite{KRSD}, see also \cite{MPV}),
or Bayou, where reconciling updates happens by redoing them in a globally determined order \cite{TTPDSH}.

The approach presented in this paper offers an improvement
over results described in \cite{NREC} in multiple ways.
We introduce a new set of commands that is symmetric and
captures more information about the updates,
and we restrict information content in directories
to exploit a further hidden symmetry in the command set.
These have the effect that reasoning about commands
becomes simpler as there are fewer edge cases,
and more powerful as 
predictions made by the reasoning are more accurate
due to the additional information content.
In fact, the new command set not only simplifies
our reconciliation algorithm, but also makes it maximal.
\cite{NREC} also lacked proofs that the update detection
and reconciliation algorithms it presented work as intended,
which we provide in the current paper.
While the results are intuitive, providing rigorous proofs
is far from trivial.
During the process, we define a number of auxiliary concepts
and show their relationships and the properties they possess.
In our view these construct a special algebraic model
that is worthy of interest and of further research on its own.

The paper is organized as follows.
We start by defining a model of filesystems and a set of commands
that we will use to describe updates and modifications in Section\nobreakspace \ref {sec_def}.
This is followed by investigating the properties and behaviors
of command pairs in Section\nobreakspace \ref {section_axioms}.
Section\nobreakspace \ref {sec_update} describes update detection and ordering filesystem commands.
The main result of this section is that under some simple conditions,
a set of filesystem commands executed in any feasible order leads
to the same changes in the filesystem.
The reconciliation and conflict detection algorithm is defined in
Section\nobreakspace \ref {sec_rec}, which merges two sets of commands.
We then proceed to prove that the output of reconciliation
is correct in the sense that 
it can be applied to the replicas without causing an error,
and it is maximal
inasmuch as no further updates can be applied under any circumstances.
In order to be able to do this, we introduce
domains of sets of command sequences,
and show that it has a number of highly convenient properties.
Sections\nobreakspace  \ref {sec_multidir} to\nobreakspace  \ref {sec_conclusion} 
describe extending our model to widen its applicability;
outline directions for further work on the introduced algebraic system,
and provide the conclusion.
We discuss related work, comparing our results to other research in Appendix\nobreakspace \ref {sec_relatedwork},
while
Appendix\nobreakspace \ref {axiom_proof} contains technical proofs for propositions on the behavior of command pairs.

\section{Definitions}\mlabel{sec_def}

\subsection{Filesystems}

We model filesystems using functions with a set of potential filesystem paths or \emph{nodes} ($\setn$) as their domain,
and a set of possible contents or values ($\setv$) as their codomain.
$\setn$  serves as a namespace for the filesystem:
it contains all possible nodes, including the ones 
where the file system contains no file or directory,
and so for most filesystems it is infinite.
Nodes in $\setn$ are also arranged in trees by a parent function ($\parent$) as described in Definition\nobreakspace \ref {def_noderel}.

\begin{mydef}[Filesystems, $\FS$]
A filesystem $\FS$ is a function
mapping the set of nodes $\setn$ to values in $\setv$:
\[ \FS: \setn \rightarrow \setv. \]
\end{mydef}

The tree-like structure of $\setn$ is determined by
an ancestor / descendant relation defined over the nodes
which arranges them in a disjoint union of rooted directed trees,
and which can be derived from the partial function $\parent$ yielding
the parent of a node in $\setn$ provided it exists.
Tao et al. \cite{TSR} describe a similar filesystem model, although
they also model inodes and restrict the filesystem to just a single tree.

\begin{mydef}[Ordering on $\setn$: $\parent$, $\descendant$, $\descendantEq$, $\unrel$]\mlabel{def_noderel}
The partial function $\parent:\setn\nrightarrow\setn$
returns the parent node of $n$,
and is undefined if $n$ is the root of a tree.

The \emph{ancestor} / \emph{descendant} relation $\descendant$ is the
strict partial ordering determined by the $\parent$ function.
We write $n\descendant m$, or $n$ is the ancestor of $m$,
iff $n=\parent^i(m)$ for some integer $i\ge 1$.
We write $n\descendantEq m$ iff $n\descendant m$ or $n=m$.

We write $n\unrel m$, or $n$ and $m$ are \emph{incomparable},
iff $n\not\descendantEq m$ and $n\not\ancestorEq m$;
that is, incomparable nodes are on different branches or on different trees.
\end{mydef}
By assumption the $\parent$ function does not induce loops, and so
$\descendant$ is indeed a strict partial ordering.

The combination of the set of nodes and the parent function,
$\langle \setn, \parent \rangle$, forms a \emph{skeleton}
which the filesystems populate with values from $\setv$.
As we require filesystem functions to be total,
we use a special value, $\empt\in\setv$, to indicate that the filesystem
is empty at a particular node, that is, there are no files or directories there.
We also assume that a filesystem has finitely many non-empty nodes,
and we consider any kind of metadata to be part of the values in $\setv$.

\begin{mydef}[Filesystem values: $\vald$, $\valfx$, $\empt$, $\eqclass{v}$]
The set of values $\setv$ is partitioned into directories, files, and the empty value $\empt$.
As usual, we write $[v]$ for the the equivalence class of $v\in\setv$ according to this partition,
which represents the type of the value.
Specific directory and file values are denoted by $\vald$ and $\valfx$, respectively,
and so the partitioning is in fact $[\vald], [\valfx], [\empt]$.
\end{mydef}

Every filesystem must have a so-called \emph{tree property}, which means that
if the filesystem is not empty at a node, and the node has a parent,
then there must be a directory at the parent node.
Using the notation introduced above, we can formally express this as follows.
\begin{mydef}[Tree property]\mlabel{def_treeprop}
A filesystem $\FS$ has the tree property iff
\[ \forall n\in\setn:
\FS(n) \neq \empt \: \Longrightarrow \: \Big[\FS\big(\parentf{n}\big)\Big] = [\vald] \]
wherever $\parent(n)$ is defined, that is, $n$ has a parent.
\end{mydef}

An essential additional assumption in our model is that there is only
a \emph{single directory value} ($|\setd|=1$).
The realization that
this creates a symmetry between directories and empty nodes
and makes it possible to exploit a hidden symmetry in filesystem commands
that allows formulating a correct and complete reconciliation algorithm
is one of the main contributions of this paper.
In Section\nobreakspace \ref {sec_multidir} we explore the applicability of our model given this assumption;
specifically,
why applicability may not be affected and ways to relax this restriction.

We also assume that there are multiple different file contents, that is, $|[\valfx]|>1$.

In the rest of paper
we fix the skeleton $\langle\setn,\parent\rangle$,
and the set of values $\setv$.
$\FS$ (with or without indices) denotes a filesystem,
and $n$, $m$ and $o$ are nodes in $\setn$.

\subsection{Commands on Filesystems}

Next, we define commands on filesystems.
As described above, we aim to select a set of commands
that captures as much information
about the operations as possible, and is also symmetric.
We start by summarizing our reasons for doing so.

Let us consider what kind of information is usually encoded in filesystem operations.
A minimal set of commands, based on the most frequent tools implemented by filesystems,
may be the following, where 
$n\in\setn$ and $\valv\in\setv$ (but $\valv\neq\empt$):
\begin{itemize}
\item \textit{create}$(n,\valv$), which creates a file or directory ($\valv$) at $n$
where the file system contains no file or directory ($\empt$);
\item \textit{edit}$(n,\valv)$, which replaces the earlier file or directory at $n$ with $\valv$;
\item \textit{remove}$(n)$, which removes the file or directory at $n$, and replaces it with $\empt$.
\end{itemize}
Regarding their output, that is, the state of the filesystem at $n$
after applying the command,
we know that after \textit{create} or \textit{edit,} $\FS(n)\neq\empt$, whereas after \textit{remove,}
$\FS(n)$ will be $\empt$. 
However, from \cite{NREC} and \cite{CBNR} we know that a useful set of axioms
will in some cases need to distinguish between, for example,
\textit{edit}s that result in directories (\textit{edit}$(n,\vald)$) and
ones that result in files (\textit{edit}$(n,\valf)$), and treat them as separate commands,
as their behaviors are quite different when combined with other commands.
Indeed, Bill Zissimopoulos' work
demonstrated \cite{BZ}
that extending this distinction to more commands ultimately simplifies
the definition of conflicting commands, as our model will then able to predict the behavior of commands
more precisely.

Notice, however, that the commands listed above also encode some information about 
their input, the state of the filesystem
before the command is applied. In particular, \textit{create}$(n,\valv)$ requires that there are no files
or directories at $n$, while \textit{edit}$(n,\valv)$ and \textit{remove}$(n)$ require the opposite.
This creates an arbitrary asymmetry where
there is now more information available about their output than about their input.
As, based on the above, we expect that encoding more information in the commands
results in a model with greater predictive powers,
and in order to resolve this asymmetry, 
we propose a set of commands that encode
the type of the input value $\FS(n)$ as well.
(Some real-life filesystem commands like \textit{rmdir} do this already.)
Because the success or failure of commands depends on types of values,
it is not necessary to encode the actual input value.

According to the above,
we model filesystem commands with partial functions 
that map filesystems to filesystems, but which
are only defined where the commands succeed.
As usual, if the function is not defined, we say that that it returns the bottom element, $\fsbroken$,
or that it \emph{breaks} the filesystem.
We describe a filesystem commands using its input, output, and
the node it is applied to.

\begin{mydef}[Filesystem commands]
Filesystem commands are partial endofunctions on filesystems,
which are not defined on filesystems where they return an error.
The commands are represented using triplets of the form
\[ \caaaa{x}{y}{n}, \]
where $x,y\in\setv$ and $n\in\setn$.
The equivalence class $[x]$ is the \emph{input type} of the function;
$y$ is its \emph{output value} implicitly specifying the \emph{output type} $[y]$,
and $n$ is the node the command is applied to.
\end{mydef}
For example, $\cbfa{n}$ represents \textit{create}$(n,\valfx)$,
and $\cdba{n}$ represents \textit{rmdir}$(n)$.
Where it is convenient, in addition to the triplets we will also use $\alpha$ and $\beta$
to denote unknown commands.

We write $\cbrk$ for the empty partial function which is not defined anywhere.
This does not naturally occur in sequences of commands we will investigate,
but is useful when reasoning about the combined effects of commands.

We note that a command or a sequence of commands is applied to a filesystem
by prefixing the command or sequence to it, for example: $\cbrk\aFS$, $\cbda{n}\aFS$, 
or $S\aFS$ if $S$ is a sequence of commands.

To define the effect of a command on a filesystem, we use
the replacement operator of the form
$\fsreplacement{\FS}{n}{\valv}$ to denote a filesystem derived from $\FS$ 
by replacing its value at $n$ with $\valv$:
\[ \fsreplacement{\FS}{n}{\valv}(m) =
   \begin{cases}
   \valv &\textrm{if~} m=n\\
   \FS(m) &\textrm{otherwise.}
   \end{cases}
\]

\begin{mydef}[Effect of commands]
The effect of the command $\cxynv$ is defined as follows:
\[
\cxynv\aFS = 
   \begin{cases}
   \fsbroken &\textrm{if $[\FS(n)]\neq[x]$,}\\
   \fsbroken &\textrm{if $\fsreplacement{\FS}{n}{y}$ violates the tree property,}\\
   \fsreplacement{\FS}{n}{y} &\textrm{otherwise.}
   \end{cases}
\]
\end{mydef}
In other words, the command $\cxynv$ breaks a filesystem if the input type $[x]$
does not match the type of the value in the filesystem at $n$, or if after replacing
the value at $n$ with $y$, the resulting filesystem ceases to satisfy the tree property
as described in Definition\nobreakspace \ref {def_treeprop}.

\myskip
There are nine groups of commands considering their input and output types.
Based on these groups we separate commands into four categories 
that reflect their overall effect:
construction commands extend the filesystem,
while destruction commands shrink it;
the replacement command replaces a file value with a file value
(which may or may not be different from the original value),
and assertion commands simply assert the type of the value at a node.
\begin{mydef}[Command categories]\mlabel{def:command_categories}
Depending on their input and output types, filesystem commands
belong to exactly one of the following four categories:
\begin{itemize}
\item[]\emph{Construction commands:}
    $\cbfa{n}$, $\cbda{n}$ and $\cfda{n}$;
\item[]\emph{Destruction commands:}
    $\cdfa{n}$, $\cdba{n}$ and $\cfba{n}$;
\item[]\emph{Assertion commands:}
    $\cbba{n}$ and $\cdda{n}$;
\item[]\emph{Replacement command:}
    $\caaa{\valfx}{\valf}{n}$.
\qedhere
\end{itemize}
\end{mydef}

\myskip
For reasons also listed in \cite{NREC}, in this model we will not consider
a \textit{move (rename)} command.
This turns out to be useful because this would be the only command that affects
filesystems at two nodes at once, therefore describing 
the dependencies for \textit{move} would call for a more complicated model.
As detailed in Section\nobreakspace \ref {sec_multidir}, even with this restriction,
an implementation of the reconciliation algorithm presented here
can still handle \emph{move} commands in its input and output
by splitting them into \emph{delete} and \emph{create} commands,
and merging them later.

\section{Command Pairs and Sequences}\mlabel{section_axioms}

So that we can describe the effects of commands independently of filesystems,
let us introduce some notation
and note some observations.
We already know that
commands usually do not occur in isolation,
and are applied to filesystems in time.
Therefore we investigate sequences of commands with a well-defined order.
\begin{mydef}[Sequences of commands and $\emptyseq$]
We concatenate commands in writing to note that they form a sequence,
and concatenate sequences of commands to form a longer sequence,
with the meaning that the resulting sequence is executed from left to right:
\[ (\alpha\cc\beta)\aFS = \beta(\alpha\aFS). \]
Sequences are also partial endofunctions on filesystems,
defined only if all commands they contain succeed in the given order.
Sequences form a monoid, and, as usual,
we write $\emptyseq$ to denote the unit element, the empty sequence,
which is defined on all filesystems and, by definition, leaves all filesystems unchanged.
\end{mydef}

\begin{mydef}[$\Dom{S}$]
For a sequence of commands $S$, $\Dom{S}$ is the domain of $S$, that is,
the set of filesystems $S$ does not break.
\end{mydef}

The following two relations 
echo the ones defined in \cite{NREC}.
In the definitions, $A,B,S$ and $T$ stand for arbitrary sequences.

\begin{mydef}[$\eqext$, $\strext$ and $\equiv$]
We write $A\eqext B$, or say that $B$ \emph{extends} $A$,
to mean that they behave in the same way
on any filesystem $A$ does not break:
$\forall \FS\in\Dom{A}:\,A\aFS=B\aFS$.
We can also see that $A\eqext B$ and $S\eqext T$ implies $A\cc S\eqext B\cc T$.
In other words, $\eqext$ is a preorder, and also a precongruence.

We write $A\equiv B$,
or say that $A$ and $B$ are \emph{equivalent,}
iff $A\eqext B$ and $B\eqext A$;
that is, $\equiv$ is the intersection of the preorder $\eqext$ with its inverse,
and so it is a congruence.

Finally, we write $A\strext B$ to mean $A\eqext B$ and $A\nequiv B$.
In particular, we write $\wrks{A}$
to mean that $\Dom{A}$ is not empty, that is, $A$ is defined on some filesystems.
\end{mydef}

It is easy to see that the equivalence $\equiv$ holds on the level of filesystems:
\begin{mylem}\mlabel{equiv_on_fs}
$A\equiv B$
iff $A$ and $B$ behave in the same way on
all filesystems, that is, $\forall \FS: A\aFS=B\aFS$.
\end{mylem}
\begin{proof}
As $A\equiv B$ means that both $A$ extends $B$ and $B$ extends $A$, it necessarily
follows that $A$ and $B$, as partial functions, are identical.
\end{proof}

\myskip
One aim of our algebraic model is to
derive as much information about the effects of sequences
of commands independently of the actual filesystems as possible.
In order to make this possible, we investigate the smallest building
blocks of sequences: pairs of commands that act on a filesystem directly one after the other.
This approach is useful as there are a limited number of command pairs,
because, as we argued above, we can disregard the exact output values of commands apart from their type,
and we can also abstract the relationship between the nodes in the two commands
to a finite number of cases.
These properties of command pairs are crucial as they determine
how a set of commands can be re-ordered to be applied to a filesystem
during synchronization, and what command pairs will never be compatible.

Command pairs in general have the form
\[ \cxynv\cc  \czwmv \]
where $x,y,z,w\in\setv$ and $n,m\in\setn$. 
Extending Definition\nobreakspace \ref {def:command_categories}, we call certain pairs
\emph{construction} or \emph{destruction pairs}.

\begin{mydef}[Construction and destruction pairs]
A pair of commands on nodes $\pn$ and $n$
is a \emph{construction pair} if the input and output types match
one of the following patterns:
   \begin{align*}
            \cbda{\pn}&\cc  \cbfa{n} \\
            \cbda{\pn}&\cc  \cbda{n} \\
            \cfda{\pn}&\cc  \caaa{\empt}{\valf}{n} \\
            \cfda{\pn}&\cc  \cbda{n}
   \end{align*}
The pair is a \emph{destruction pair} if the types match one of the following:
   \begin{align*}
            \cfba{n}&\cc  \cdba{\pn} \\
            \cfba{n}&\cc  \caaa{\vald}{\valf}{\pn} \\
            \cdba{n}&\cc  \cdba{\pn} \\
            \cdba{n}&\cc  \cdfa{\pn} \qedhere
   \end{align*}
\end{mydef}
We can see that construction pairs consist of construction commands,
while destruction pairs consist of destruction commands.

\MakeUppercase Lemma\nobreakspace \ref {rules_lemma} summarizes the basic properties of command pairs.
The claims listed are named \emph{Rules}
as they can also be interpreted as inference rules in a pure
algebraic treatment of filesystem synchronization.
Section\nobreakspace \ref {sec_algebra} elaborates this approach.

\begin{mylem}\mlabel{rules_lemma}
\newcounter{rulecounter}
\begin{list}{\bf Rule~\arabic{rulecounter}.}{\usecounter{rulecounter}}

\item[] %

\item\mlabel{ax_separate_commute}
\axaxseparatecommute

\item\mlabel{ax_separate_nobreaks}
\axaxseparatenobreaks

\item\mlabel{ax_same_breaks}
\axaxsamebreaks

\item\mlabel{ax_same_emptyseq}
\axaxsameemptyseq

\item\mlabel{ax_same_singlec}
\axaxsamesinglec

\item\mlabel{ax_directchild_breaks}
\axaxdirectchildbreaks

\item\mlabel{ax_directparent_breaks}
\axaxdirectparentbreaks

\item\mlabel{ax_distantrel_breaks}
\axaxdistantrelbreaks

\item\mlabel{ax_child_assert}
\axaxchildassert

\item\mlabel{ax_parent_assert}
\axaxparentassert

\item\mlabel{ax_assert}
\axaxassert

\end{list}
\end{mylem}

We have included Rules\nobreakspace \ref {ax_directchild_breaks} and\nobreakspace  \ref {ax_directparent_breaks}
for completeness, but we will use the more generic Lemma\nobreakspace \ref {simple_distant_pairs},
which extends these rules to non-adjacent commands, although only in the case
of command sequences without any assertion commands.
The Rules can be derived from the filesystem model
relatively easily. Proofs are included in Appendix\nobreakspace \ref {axiom_proof}.

\myskip
We also define the concept of two commands being independent.

\begin{mydef}[$A\indep B$: Independent commands, sequences and sets of commands]\mlabel{def_indep}
Two commands $\alpha$ and $\beta$ 
are independent, written as $\alpha\indep\beta$, if 
they commute and do not break all filesystems:
\[ \alpha\cc\beta \equiv \wrksx{\beta\cc\alpha}. \]
For two sequences or unordered sets of commands $A$ and $B$ we write $A\indep B$ if
for all $\alpha$ in $A$ and all $\beta$ in $B$, $\alpha\indep\beta$.
We also write $\alpha\indep B$ for $\{\alpha\}\indep B$.
\end{mydef}

It is intentional that we use the same symbol for independent commands
as for incomparable nodes. As the
Corollary below
shows, these concepts are closely related.

\begin{mycor}\mlabel{incomparable_is_independent}
If $\cxynv$ and $\czwmv$ are different commands,
and none of them is an assertion command, then
$\cxynv\indep\czwmv$ if and only if $n\unrel m$.
\end{mycor}

\MakeUppercase Corollary\nobreakspace \ref {incomparable_is_independent} follows from
the next Lemma,
which characterizes independent commands in a more detailed way.

\begin{mylem}\mlabel{independent_details}
For two commands $\cxynv\indep\czwmv$ if and only if
one of the following holds:
\begin{itemize}
\item $n\unrel m$,
\item $n=m$ and the commands are equivalent assertion or replacement commands,
\item either $n\descendant m$ or $m\descendant n$,
\emph{and} either the command on the ancestor asserts a directory,
or the command on the descendant asserts an empty node (or both).
\end{itemize}
\end{mylem}

\begin{proof}
By definition we know that $\cxynv\cc\czwmv\equiv\wrksx{\czwmv\cc\cxynv}$.
The first case is a restatement of
Rules\nobreakspace \ref {ax_separate_commute} and\nobreakspace  \ref {ax_separate_nobreaks}.
Otherwise $n\nunrel m$, in which case either $n=m$, or, without loss of generality, $n\descendant m$.

If $n=m$, then from Rule\nobreakspace \ref {ax_same_breaks} we must have $[y]=[z]$ and $[w]=[x]$
as otherwise the commands, in one order or the other, would break all filesystems.
The equivalence also implies $[x]=[z]$ as otherwise the two sides would not be defined on the same filesystem,
as well as $y=w$ as their effects is the same.
The two commands are therefore the same, and they are either assertion or replacement commands.

The remaining case is that $n\descendant m$.
If the nodes are not directly related, then Rule\nobreakspace \ref {ax_distantrel_breaks} gives
the conditions in the third case in the lemma.
If they are directly related, 
Rules\nobreakspace \ref {ax_directchild_breaks} and\nobreakspace  \ref {ax_directparent_breaks}
gives the same conditions, while they also allow the commands being
a construction or destruction pair.
However, the reversals of construction and destruction pairs break all filesystems.
\end{proof}

While equivalent commands can hardly be considered independent,
we retain this name for this relation as in most
cases we use it on pairs of different commands that are also not assertion commands.

\section{Update Detection}\mlabel{sec_update}

In a command-based reconciliation solution we assume that we have two sequences of commands
$A$ and $B$ that have modified a single filesystem $\FS$ 
yielding two different replicas $\FS_1$ and $\FS_2$ which we
need to reconcile. While it is conceivable that the sequences are based on a record of
all operations that modified the two filesystems, in most filesystem implementations
such records do not exist, and therefore we must construct suitable sequences
by comparing $\FS_1$ (and $\FS_2$) to their common ancestor, $\FS$. 
This is called \emph{update detection.}

The set of commands $U$ necessary to transform $\FS$ into $\FS_1$ can be collected
by inspecting the (finite) union of non-empty nodes in the two filesystems.
If at node $n$, $\FS(n)=x$ and $\FS_1(n)=y\neq x$, then we add the command $\cxyaa{n}$ to $U$.
We can always do so as there is a command available for all combinations of input and output types and values.
We see that any suitable $U$ necessarily contains commands on all nodes at which the values have changed.

\begin{mydef}[Minimal and simple command sets and sequences]\mlabel{def_min_simp}
A sequence or set of commands is \emph{minimal} if it contains at most one command on each node.
It is \emph{simple} if it is minimal and it does not contain assertion commands.
\end{mydef}

This update detector therefore yields a simple set of commands because we only add a single command
for each node, and we only add commands that are necessary, that is, there will be no 
assertion commands in the set.

The next step in generating the sequences is to order the commands collected.
As this task is at the heart of reconciliation itself independently of update detection,
we discuss it in the next section.
Then, in Theorem\nobreakspace \ref {update_works}, we prove that the resulting sequence 
returned by the update detector actually works without breaking the filesystem.

\subsection{Ordering commands}\mlabel{ordering}

We often encounter the case where we only have a set of commands without a specified order.
As we have seen above, this can occur after the first stage of update detection,
but, more importantly,
it is actually the task of the reconciler to determine whether there is an order,
and if yes, what order,
in which updates from different replicas can be applied to a filesystem.

As multiple orders may be possible,
we describe our ordering algorithm by
defining a subset of the permutations of the commands, $\ordersetsign$.
\MakeUppercase Lemma\nobreakspace \ref {simple_reorder_equiv} proves that $\ordersetsign$ indeed contains all possible
sequences of commands that do not break all filesystems.
The facts that all valid reorderings of simple sequences are equivalent,
and all equivalent sequences are reorderings follow from this Lemma, and are important 
properties of simple sequences.

We define $\ordersetsign$ with the help of a partial order.
\begin{mydef}[$\orderrel$]\mlabel{def_orderrel}
The binary relation $\orderrel$ holds for some command pairs
on directly related nodes in the following cases:
\begin{itemize}
\item $\caaa{x}{y}{\parentf{n}} \orderrel \caaa{z}{w}{n}$ if both commands are construction commands, and
\item $\caaa{x}{y}{n} \orderrel \caaa{z}{w}{\parentf{n}}$ if both commands are destruction commands. \qedhere
\end{itemize}
\end{mydef}
This relation is clearly irreflexive and anti-symmetric,
and as a command cannot be both a construction and a destruction command and
the relation follows the tree of nodes, it is acyclic.
Therefore its transitive closure, also denoted by $\orderrel$, is a partial order.

\begin{mydef}[$\ordersetsign$]\mlabel{def_orderset}
For a simple sequence of commands $\wrksx{S}$,
$\orderset{S}$ is the set of permutations of the commands in $S$ respecting the partial order defined by $\orderrel$,
that is,
for any $T\in\orderset{S}$ and any $\alpha,\beta\in T$, if $\alpha\orderrel\beta$ then $\alpha$
precedes $\beta$ in $T$.
\end{mydef}

Returning a valid $T\in\orderset{S}$ order for a set of commands $S$
can be implemented using any well-known algorithm for topological sorting,
where the directed acyclic graph that is the input of the algorithm is defined by $\orderrel$.

The principal ideas behind the definition of $\orderrel$ are 
Rule\nobreakspace \ref {ax_separate_commute}, from which we know that commands on incomparable
nodes can be applied in any order, so we only need to focus on the order of commands
with comparable nodes;
Lemma\nobreakspace \ref {connected_changes}, from which we know that if a simple sequence contains commands on two comparable
nodes, then it contains commands on all nodes in between, so it is enough to specify
the order of commands on directly related nodes;
and finally, Lemma\nobreakspace \ref {simple_distant_pairs}, from which we know
that commands on parent--child pairs must be construction or destruction pairs.

To prove Lemma\nobreakspace \ref {simple_reorder_equiv}, we need the following simple results.

\begin{mycor}\mlabel{subseq_in_orderset}
For simple sequences $S$ and $T$,
if $T\in\orderset{S}$, then for any subsequence of $S$, $S_0$,
and for the corresponding subsequence of $T$, $T_0$, which
contains the same commands but in the order they are found in $T$,
$T_0\in\orderset{S_0}$ must hold.
\end{mycor}

This follows directly from Definition\nobreakspace \ref {def_orderset}.

\begin{mylem}\mlabel{simple_distant_pairs}
For any simple sequence $\wrksx{S}$ and its subsequence $\cxynv\cc\allowbreak\czwmv$,
if $n=\parent(m)$, then the subsequence is a construction pair;
and if $\parent(n)=m$, then the subsequence is a destruction pair.
\end{mylem}

Informally speaking, the proposition means that
the two commands must form a construction or destruction pair even
if they are not next to each other.
This is true because in simple sequences
there cannot be another command on $n$ or $m$,
so if the commands are incompatible, no command between them can change that.
In fact, this Lemma is the general case
of Rules\nobreakspace \ref {ax_directchild_breaks} and\nobreakspace  \ref {ax_directparent_breaks} when there are no assertion commands.

\begin{proof}
Formally, we prove our proposition by contradiction,
and in this version of the proof we reach back to our filesystem model.
By assumption $\wrks{S}$ and therefore there is a $\FS$ for which $S\aFS\neq\fsbroken$.
Select two commands on a node and its parent, $\cxynv$ and $\czwmv$,
and split $S$ around them into three parts:
\[ S = S_0 \cc \cxynv \cc S_1 \cc \czwmv \cc S_2, \]
where any of the three parts can be empty.
As $S$ is simple, there are no commands on $n$ or $m$ in $S_0$, $S_1$, or $S_2$,
and therefore $(S_0\cc\cxynv\cc S_1)\aFS(n) = y$,
and $[(S_0\cc\cxynv\cc S_1)\aFS(m)] = [S_0\aFS(m)]$,
and so $[z]$ must be $[S_0\aFS(m)]$.

If $n=\parent(m)$, then 
we know $(S_0\cc\cxynv\cc S_1)\aFS(n) = y$, and as $[z]\neq[w]$, either
$[(S_0\cc\cxynv\cc S_1)\allowbreak\aFS(m)]=[z]\neq[\empt]$,
or $[(S_0\cc\cxynv\cc S_1\cc\czwmv)\allowbreak\aFS(m)]=[w]\neq[\empt]$.
Therefore $y$ must be $\vald$, as otherwise the tree property would be violated
when applying $\czwmv$,
and $[x]\neq[\vald]$ as $\cxynv$ is not an assertion command.
As $[S_0\aFS(n)]=[x]\neq[\vald]$, we know
$S_0\aFS(m)=\empt$ as otherwise the tree property would be violated.
Therefore $[S_0\aFS(m)]=[z]=[\empt]$, which, combined with $y=\vald$
means that $\cxynv\cc\czwmv$ is a construction pair.

If $\parent(n)=m$, then 
as $(S_0\cc\cxynv\cc S_1)\aFS(n) = y$, and either
$[(S_0\cc\cxynv\cc S_1)\allowbreak\aFS(m)]=[z]\neq[\vald]$
or $[(S_0\cc\cxynv\cc S_1\cc\czwmv)\allowbreak\aFS(m)]=[w]\neq[\vald]$,
$y$ must be $\empt$, and $[x]\neq[\empt]$.
As $[S_0\aFS(n)]=[x]\neq[\empt]$, we know $S_0\aFS(m)=\vald$.
Therefore $[S_0\aFS(m)]=[z]=[\vald]$, which, combined with $y=\empt$
means that $\cxynv\cc\czwmv$ is a destruction pair.
\end{proof}

\begin{mycor}\mlabel{order_is_only_possible}
For any simple sequence $\wrksx{S}$ and its permutation $S'$, 
\[ S'\not\in\orderset{S} \Longrightarrow S'\equiv\cbrk. \]
\end{mycor}
\begin{proof}
This follows directly from Lemma\nobreakspace \ref {simple_distant_pairs}.
If $S'\not\in\orderset{S}$ then by Definitions\nobreakspace \ref {def_orderrel} and\nobreakspace  \ref {def_orderset}
$S'$ contains two commands on a node $n$ and its parent
that, as a subsequence,
form a construction or destruction pair in the wrong order.
\MakeUppercase Lemma\nobreakspace \ref {simple_distant_pairs} implies that no such subsequence exists
if $\wrksx{S'}$, and therefore $S'\equiv\cbrk$.
\end{proof}

\begin{mylem}\mlabel{connected_changes}
Given a set of commands that
does not contain assertion commands,
and which can be applied to a filesystem in some order without breaking it,
and which contains commands on $\nn$ and $n$ where $\nn\descendant n$,
then the set must also contain a command
on each node between $\nn$ and $n$.
\end{mylem}
\begin{proof}
Without loss of generality, we can assume that $\nn\neq\parent(n)$.
We prove that under the given conditions, the set must contain a command on $\parent(n)$.
Then, by reapplying this result, we know that the set must contain commands on every
ancestor of $n$ up to $\nn$.

Furthermore,
we prove this proposition for sequences, not sets, as if all sequences must contain a command on $\parent(n)$,
then so must all sets because otherwise there would be no order in which the commands they contain could be
applied to a filesystem.

Let $A$ be a sequence that satisfies the conditions;
we therefore know that it contains a command $\cxynv$ on $n$
and another command $\czwnnv$ on $\nn$.
By contradiction assume that there are no commands on $\parent(n)$ in $A$.
Next, we create a new sequence $A'\equiv A$ in which $\cxynv$ and $\czwnnv$ are next to each other.
If they are already next to each other in $A$, there is nothing to do.
Otherwise, consider the command next to $\cxynv$ in the direction where $\czwnnv$ is.
Let this command be $\cqrmv$.
If $\czwnnv$ is to the right, then $A$ looks like the following:
\[ A = \cdots\cc\cxynv\cc\cqrmv\cc\cdots\cc\czwnnv\cc\cdots \]
If $n\unrel m$, then swap $\cxynv$ and $\cqrmv$. Based on Rule\nobreakspace \ref {ax_separate_commute} we know that the new
sequence is equivalent to $A$.
Otherwise, we know $m\neq\parent(n)$ as there are no commands on $\parent(n)$, and so
from Rule\nobreakspace \ref {ax_distantrel_breaks} we get $A\equiv\cbrk$ which contradicts our assumptions.
(Note that $A$ does not contain assertion commands.)
By repeating this step we can therefore convert $A$ into $A'$ where $\cxynv$ and $\czwnnv$ are neighboring commands.
However, then Rule\nobreakspace \ref {ax_distantrel_breaks} applies to the sequence and therefore $A\equiv A'\equiv\cbrk$ which
is again a contradiction.
\end{proof}

\begin{mylem}\mlabel{equiv_simple_same_commands}
If two simple sequences $A$ and $B$ are equivalent
and do not break all filesystems ($A\equiv \wrksx{B}$),
then they must contain the same commands.
\end{mylem}
\begin{proof}
Let $A$ and $B$ be simple sequences
such that $A\equiv \wrksx{B}$,
and $\FS$ be a filesystem that they are defined on.
We use a proof by contradiction, and assume that they do not contain the same commands.
Without loss of generality, we can assume
that $A$
contains $\cxynv$, and $B$ either contains a different command
$\czwnv$ on $n$, or no command on $n$ at all.
As $A$ is simple, we know that $\cxynv$ is not an assertion command,
and therefore either $[x]\neq [y]$, or it is a replacement command.

If $[x]\neq [y]$, then $\FS(n)\neq y$ as $[\FS(n)]=[x]$.
Therefore, if $B$ has no command on $n$, then $B\aFS(n)=\FS(n)\neq y=A\aFS(n)$ and
$A$ and $B$ cannot be equivalent.
If $B$ includes $\czwnv$, then we know that $[z]=[x]$ as $B$ does not break $\FS$ either,
and that $w=y$ as $B\aFS(n)$ must
be $y$ and $B$ only has one command on $n$.
This means that $\czwnv=\cxynv$, which is a contradiction.

If $\cxynv$ is a replacement command, then we know $[x]=[y]=[\valfx]$, and
$[\FS(n)]=[\valfx]$. If $\FS(n)=y$, then instead of this filesystem,
consider $\fsreplacement{\FS}{n}{\valf}$ where $\valf$ is any file value
other than $y$.
From here the proof concludes in the same way as above.
\end{proof}

\begin{mylem}\mlabel{simple_reorder_equiv}
For a simple sequence $\wrksx{S}$,
$\orderset{S}$ is the set of all simple sequences equivalent to $S$.
\end{mylem}

\begin{proof}
First, we prove, by contradiction, that if $T$ is a simple sequence and $T\equiv S$, then $T\in\orderset{S}$.
Assume $T\equiv S$ but $T\not\in\orderset{S}$.
Then, from Lemma\nobreakspace \ref {equiv_simple_same_commands} we know $T$ is a permutation of $S$,
and from Corollary\nobreakspace \ref {order_is_only_possible} we know $T\equiv\cbrk$ which
is a contradiction as $T\equiv \wrksx{S}$.

Next, we prove that if $T\in\orderset{S}$, then $T\equiv S$.
We proceed by induction on the length of $S$.
Our base cases are $\emptyseq$, when $\orderset{\emptyseq}=\{\emptyseq\}$,
and one-long sequences, when this is trivially true.
In our induction step we assume that $T^*\in\orderset{S^*}\Rightarrow T^*\equiv S^*$ holds
for all sequences of length $i$ or less, where $i\geq 1$.

Let us consider $T\in\orderset{S}$ where $T$ and $S$ are of length $i+1$.
Let the first command in $S$ be $\cxynv$, that is, $S=\cxynv\cc S_0$.
We proceed by a nested induction on $j$, the number of commands before $\cxynv$ in $T$.
Our base case is $j=0$, when it is the first command, which means $T=\cxynv\cc T_0$.
Then from Corollary\nobreakspace \ref {subseq_in_orderset} we know $T_0\in\orderset{S_0}$
and from the induction hypothesis $T_0\equiv S_0$ from which we get $T\equiv S$.

In our nested induction step we assume $T'\equiv S$ for all $T'\in\orderset{S}$
where there are at most $j$ commands in $T'$ to the left of $\cxynv$.
Let us then consider $T$ where there are $j+1$ such commands.
We aim to transform $T$ into $T''\equiv T$ by swapping $\cxynv$ with the preceding
command. Then from the induction hypothesis we know $T\equiv T''\equiv S$,
which proves our lemma.

Let the command to the left of $\cxynv$ in $T$ be $\czwmv$.
As $S$ is simple, we know $n\neq m$.
If $n\unrel m$, then from Rule\nobreakspace \ref {ax_separate_commute} we know we can swap the two commands
and get an equivalent sequence.
We finish the proof by showing that $n\unrel m$ must hold, as
all other cases lead to contradiction.

\newcommand{\indx}{\varphi}
If $n\nunrel m$, then either $n\descendant m$, or $m\descendant n$, and so from
Lemma\nobreakspace \ref {connected_changes} we know that $S$ contains a command on all nodes
between $n$ and $m$, and therefore so does $T$.
Let these nodes be $n_0, n_1, \ldots, n_k$ where $n_0=n$ and $n_k=m$,
and where 
either $n_\indx=\parentf{n_{\indx+1}}$ for all $\indx$ between $1$ and $k-1$,
or $n_{\indx+1}=\parentf{n_\indx}$ for all $\indx$.
We assume that $n_\indx=\parentf{n_{\indx+1}}$ is true.

Also, let $I^S_\indx$ be the index of the command on $n_\indx$ in $S$.
By Lemma\nobreakspace \ref {simple_distant_pairs}, for all $\indx$, both the commands on $n_\indx$ and on $n_{\indx+1}$ must be either
construction commands, or destruction commands, depending on
whether $I^S_\indx<I^S_{\indx+1}$ or $I^S_{\indx+1}<I^S_\indx$.

As a command cannot be both a construction and a destruction command,
this means that the indices $I^S$ are either monotone increasing, or monotone
decreasing, but cannot change direction.
And as $S$ begins with $\cxynv$ and so $I^S_0=1$, we know they are monotone increasing,
and the commands on $n_\indx$ are construction commands.

Let us turn to the location of these commands in $T$,
and denote their indices in $T$ with $I^T_\indx$.
We know $I^T_k=I^T_0-1$ as $\czwmv$ precedes $\cxynv$.
This means there must exist a $1\leq \indx<k$ for which $I^T_{\indx+1}<I^T_\indx$.
We therefore know that
$\caaa{x_{\indx+1}}{y_{\indx+1}}{n_{\indx+1}} \cc \caaa{x_\indx}{y_\indx}{n_\indx}$
is a subsequence of $T$.
However, because we know these are construction commands, and by assumption $\parentf{n_{\indx+1}}=n_\indx$,
the inverse of $\orderrel$ holds for this subsequence, which is a contradiction.

It can be shown in the same way
that the case when the $n_\indx$ nodes ascend the tree,
that is, $n_{\indx+1}=\parentf{n_\indx}$ is true, also leads to a contradiction.
In this case we find that all commands are destruction commands,
and a subsequence in $T$ again violates $\orderrel$.
\end{proof}

It follows from Lemma\nobreakspace \ref {simple_reorder_equiv} that given a simple set of commands $S$,
which can be applied to a filesystem in some order without breaking it,
this order will be in $\orderset{S}$,
and all sequences in $\orderset{S}$
represent the same partial endofunction on filesystems.
This allows us to treat such a set as a sequence for any purpose
where the internal order of the commands is irrelevant.
\begin{mydef}[$\ordp{S}$]
For a simple set or sequence of commands $\wrksx{S}$,
$\ordp{S}$ is the unique partial endofunction on filesystems
defined by any sequence in $\orderset{S}$.
\end{mydef}

Building on the results above, we finish by proving an important result,
which states that command sequences can be made simple using syntactical operations only.

\begin{myth}\mlabel{can_simplify}
Each sequence of commands $\wrksx{S}$
can be transformed into a simple sequence
$S^*\eqnrw S$ by applying repeatedly the following transformations:
{\rm(i)} swap neighboring commands on incomparable nodes;
{\rm(ii)} simplify neighboring commands on the same node;
{\rm(iii)} delete assertion commands.
\end{myth}
\begin{proof}
We can remove all assertion commands, as based on
Rule\nobreakspace \ref {ax_assert} the resulting sequence extends the original one.
We therefore assume that $S$ contains no assertion commands.

\newcommand{\ucx}{\alpha}
\newcommand{\ucy}{\beta}
\newcommand{\ucz}{\gamma}

We aim to find a simple sequence that extends $S$.
If $S$ is already simple, there is nothing to do.
If it is not simple, let $\ucx$ be the first command
on a node $n$ that is already present in $S$,
and let $\ucy$ be the previous command on $n$.
Splitting $S$ around these commands we get
\[ S = S_0 \cc \ucy \cc S_1 \cc \ucx \cc S_2. \]

We now show that $S_1\indep\ucx$,
which is equivalent to $S_1\indep\ucy$ as both commands are on $n$.
We use an inverse proof and assume that there is a command on $m$ in $S_1$
where $m\descendant n$ or $m\ancestor n$.
We know both $\ucy\cc S_1$ and $S_1\cc\ucx$ are simple,
and both are defined on a filesystem
(on $S_0\aFS$ and $(S_0\cc\beta)\FS$, respectively).
Therefore from Lemma\nobreakspace \ref {connected_changes} we know that
$S_1$ must also contain a command on $m'$ where
$m'=\parentf{n}$ or $n=\parentf{m'}$.
Let this command be $\ucz$.

If $n=\parentf{m'}$,
from Lemma\nobreakspace \ref {simple_distant_pairs} 
and the relationship of nodes in construction and destruction pairs
we know that the subsequence
$\ucy\cc\ucz$ must be a construction pair,
and both $\ucy$ and $\ucz$ must be construction commands.
But we also know that $\ucz\cc\ucx$ must be a destruction pair,
and both $\ucz$ and $\ucx$ must be destruction commands.
This is a contradiction as a command cannot be both.
We arrive at the same contradiction when $m'=\parentf{n}$.

We therefore know that $S_1\indep\ucx$ and so
from Rule\nobreakspace \ref {ax_separate_commute}
\[ S \equiv S_0\cc\ucy\cc\ucx\cc S_1 \cc S_2. \]
We can apply Rules\nobreakspace \ref {ax_same_emptyseq} and\nobreakspace  \ref {ax_same_singlec} to
$\ucy\cc\ucx$ and get a $S^*\eqext S$ that contains one or two less commands on $n$.
Repeating the swap and the simplification we
can convert $S$ into a simple sequence.
\end{proof}

\subsection{The Correctness of Update Detection}

With the relation $\orderrel$ we can now define our update detection algorithm.
Its inputs are the original filesystem $\FS$ which has been modified to yield $\FS^*$.
\begin{mydef}[Update detection]\mlabel{def_upddetect}~
\begin{enumerate}
    \item For each node $n$ where the value in $\FS$ and $\FS^*$ differ, add the command $\caaa{\FS(n)}{\FS^*(n)}{n}$ to
        a set of commands $U^*$. The result is a simple set of commands.
    \item Order the commands according to $\orderrel$, that is, 
        return any sequence $U$ from $\orderset{U^*}$. \qedhere
\end{enumerate}
\end{mydef}

\MakeUppercase Theorem\nobreakspace \ref {update_works} below proves that $U$ functions as expected, that is, $U\aFS=\FS^*$.
While this is trivial if we know that $U$ does not break $\FS$, we
still need to show that it is defined on $\FS$.

On systems where filesystem updates are recorded,
an alternative update detection algorithm is to
omit step 1, and use Theorem\nobreakspace \ref {can_simplify} to simplify the series of recorded updates
into a simple sequence containing the necessary updates.

\begin{myth}\mlabel{update_works}
For a simple sequence of commands $U$ returned by the update detector
when comparing the non-broken $\FS^*$ to the original $\FS$,
$U\aFS = \FS^*$.
\end{myth}

\begin{proof}
First, we create a sequence of commands $S$ for which $S\aFS=\FS^*$.
We create a topological sort of all nodes in $\setn$ that are
not empty in $\FS$ so that children come first,
that is, if $n'\descendant n$, then $n$ precedes $n'$.
For each node in the given order, we add a destruction command to a sequence $S_0$
that deletes the value at $n$.
We know that the commands will never break the filesystem, as the types of nodes
with empty children can be freely changed, and that $S_0\aFS$ is empty at all nodes.
Next, we order nodes in $\setn$ that are not empty in $\FS^*$ in the reverse order,
that is, if $n'\descendant n$, then $n'$ precedes $n$.
We add a construction command to a sequence $S_1$ for each node that changes an empty value
to $\FS^*(n)$.
We know that $S_1$ is defined on an empty filesystem for the same reason,
and clearly $(S_0\cc S_1)\aFS=\FS^*$.
By Theorem\nobreakspace \ref {can_simplify} there exists a simple $S$ for which $S\eqnrw (S_0\cc S_1)$.

Therefore we know $S\aFS=\FS^*$, and that
$\FS^*$ and $\FS$ differ at exactly those nodes that $S$ has commands on.
$U$ must also contain commands on the same set of nodes,
and it must contain the same commands,
as for the command $\cxynv$, $[x]$ must be $[\FS(n)]$, and $y$ must be $\FS^*(n)$.
From this we know that $U$ is a permutation of $S$ and so $\orderset{U}=\orderset{S}$.

As the update detector uses the $\orderrel$ relation to order the commands in $U$, trivially $U\in\orderset{U}$,
which means $U\in\orderset{S}$,
and so from Lemma\nobreakspace \ref {simple_reorder_equiv} we know that $U\equiv S$,
from which $U\aFS=\FS^*$.
\end{proof}

\section{Reconciliation}\mlabel{sec_rec}

If we start with two copies of a filesystem $\FS$,
and two different sequences are applied to the copies to yield $\FS_A=A\aFS$
and $\FS_B=B\aFS$, then our aim is to define sequences of commands $\reca$ and $\recb$
so that $\recb\aFS_A$ and $\reca\aFS_B$ would be as close to each other as possible.

We work based on the assumption that to achieve this, we need
to apply to $\FS_B$ the commands that have been applied to $\FS_A$, and \emph{vice versa}.
As some commands may have been applied to both filesystems, our first approximation
is to apply the commands in $\bmanp$ to $\FS_A$ and those in $\ambnp$ to $\FS_B$.
This, however, will break both filesystems if there have been incompatible updates
in $A$ and $B$. 
Our aim is therefore to provide an algorithm that selects the commands 
$\reca \subset A\setminus B$
and $\recb \subset B\setminus A$ 
so that $\reca\aFS_B\neq\fsbroken$ and $\recb\aFS_A\neq\fsbroken$,
and show that these are the longest sequences with this property.

We will work based on update sequences $A$ and $B$ that are simple.
If the reconciliation algorithm is to work on arbitrary sequences,
Theorem\nobreakspace \ref {can_simplify} can be used to convert these sequences to simple ones.

\begin{mydef}[Reconciliation]\mlabel{def:reconciliation}
For two simple sequences $A$ and $B$,
$\reca$ is a valid ordering of
the largest subset of $A\setminus B$
that is independent of $B\setminus A$:
\[ \reca = \ords{\alpha \whr \alpha\in \ambnp  \mbox{~and~}  \alpha\indep \bmanp }. \]
$\recb$ can be obtained in a similar way by reversing $A$ and $B$
in the definition.
\end{mydef}

Note that $\alpha\in\reca$ is not an assertion command, and $\bmanp$ contains no assertion commands,
and $\alpha\not\in\bmanp$, therefore $\alpha\indep\bmanp$ is equivalent to
the node of $\alpha$ being incomparable to the nodes of the commands in $\bmanp$.
Therefore whether two commands are independent 
for the purposes of reconciliation
is easy to determine programmatically.
Then, to generate the sequence $\reca$ from the resulting set, 
apply topological sorting according to $\orderrel$ as defined in Section\nobreakspace \ref {ordering}.

\myskip
We aim prove that $\reca$ and $\recb$ can be applied to the replicas,
that is, without loss of generality,
\[ \reca\aFS_B\neq\fsbroken. \]
So that we can formalize statements needed to prove this,
we introduce the partial order $\worksmeqsign$ that describes conditions under which
sequences of commands work, that is, do not break a filesystem.

\subsection{Domains of Sets of Command Sequences}

\begin{mydef}[Sets of sequences, their domains]
As we will frequently refer to sets of sequences,
we will use calligraphic letters (e.g. $\sqs{A}$, $\sqs{B}$, $\sqs{C}$ and $\sqs{S}$)
to denote such sets for brevity.
We will write $\Dom{\sqs{A}}$ for $\bigcap_{A\in\sqs{A}} \Dom{A}$,
and $\sqs{A}\cc\sqs{B}$ for $\{A\cc B\whr A\in\sqs{A}, B\in\sqs{B}\}$.
\end{mydef}

\begin{mydef}[$\worksmeqsign$]\mlabel{def_works}
For two sets of sequences $\sqs{A}$ and $\sqs{B}$
we write $\worksm{\sqs{A}}{\sqs{B}}$ iff $\Dom{\sqs{A}} \subseteq \Dom{\sqs{B}}$;
that is, iff all sequences in $\sqs{B}$ are defined on (do not break)
all filesystems on which sequences in $\sqs{A}$ are defined.
We write $\sqs{A}\worksmsign\sqs{B}$ iff $\Dom{\sqs{A}} \subset \Dom{\sqs{B}}$.
\end{mydef}

When $\sqs{A}$ or $\sqs{B}$ contains a single sequence,
we leave out the curly brackets and write,
e.g. $A\cc\sqs{B}$ to mean $\{A\}\cc\sqs{B}$,
or $\worksm{A}{B}$ to mean $\worksmbb{A}{B}$.
Also, we write $\worksmnil{\sqs{A}}$ to mean that
there is at least one filesystem on which all sequences in $\sqs{A}$ are defined.
If $\sqs{A}$ contains a single sequence only, $A$, this is equivalent to $\wrks{A}$.

We can see that $A\eqext B$ implies $\worksm{A}{B}$, as the latter
only requires that $B$ is defined where $A$ is defined, 
while the former also requires
that where they are defined, their effect is the same.

The following claims follow from the definition:

\begin{myclm}\mlabel{worksextpostfix}
$\forall A,S: \worksm{A\cc S}{A}$, that is, if a sequence is defined,
its initial segments are also defined.
\end{myclm}

\begin{myclm}\mlabel{works_restricted}
$\forall \sqs{A},\sqs{B},S: \worksm{\sqs{A}}{\sqs{B}} \Rightarrow \worksm{S\cc \sqs{A}}{S\cc \sqs{B}}$,
that is,
the relationship between the domains of the sets $\sqs{A}$ and $\sqs{B}$ 
does not change if both are prefixed by a sequence $S$.
\end{myclm}
\begin{proof}
This is because the sequence $S$, as an endofunction on filesystems, is a binary relation,
and we can treat its inverse relation as a one-to-many mapping between filesystems
that maps $\Dom{\sqs{A}}$ to $\Dom{S\cc\sqs{A}}$ and $\Dom{\sqs{B}}$ to $\Dom{S\cc\sqs{B}}$.
As such a one-to-many mapping maps a subset of a set to a subset of the image of the set,
if $\Dom{\sqs{A}}\subseteq\Dom{\sqs{B}}$, then $\Dom{S\cc\sqs{A}}\subseteq\Dom{S\cc\sqs{B}}$.
\end{proof}

We proceed by proving the following lemmas.

\begin{mylem}\mlabel{combine_independent_commands}
The combination of independent commands is defined on all filesystems
where both of the commands are defined:
\[ \alpha\indep \beta \Rightarrow \worksmbx{\alpha, \beta}{\alpha\cc \beta}. \]
\end{mylem}
\begin{proof}
The proposition uses Lemma\nobreakspace \ref {independent_details},
and we follow the three cases listed there.
Let $\alpha=\cxynv$ and $\beta=\czwmv$.

The first case is that $n\unrel m$ where $\alpha=\cxynv$, $\beta=\czwmv$.
Assume, by contradiction, that for a filesystem $\FS$
both $\cxynv\aFS$ and $\czwmv\aFS$ are defined, but $(\cxynv\cc\czwmv)\aFS$ is broken.
Since $n\neq m$, we know the input types of the command are compatible with $\FS$,
so $(\cxynv\cc\czwmv)\aFS$ must be broken because it violates the tree property.

We know that applying $\czwmv$ breaks $\cxynv\aFS$. As it changes the filesystem
at $m$, this can be because either $m$ would no longer be a directory but has non-empty children,
or $\parentf{m}$ is not a directory but acquired non-empty children.
Let us look at the first case.
Since $n\indep m$, we know the parent and children of $m$ have the same value
in $\FS$ as in $\cxynv\aFS$. However, this leads to contradiction
as applying $\czwmv$ to $\cxynv\aFS$ leads to a broken filesystem, but
applying it to $\FS$ does not.

Therefore $\parentf{m}$ must not be a directory in $\FS$ but acquire non-empty children.
Since $\czwmv\aFS$ is defined, it must be $\cxynv$ that changes the children of
$\parentf{m}$ in this way, which means that $n$ and $m$ must be siblings.
As then the value at $\parentf{n}=\parentf{m}$ is not changed by either command,
we know it cannot be a directory in $\FS$, either, and that therefore
$\FS$ is empty at all children of the parent node.
This means that $[x]=[z]=[\empt]$, and
since both $\cxynv\aFS$ and $\czwmv\aFS$ are defined, $y=w=\empt$ must also hold.
However, this contradicts our assumption that $(\cxynv\cc\czwmv)\FS$ is broken.

The second case is that the two commands are the same and are assertion or replacement
commands, when the proposition is trivially true.
Finally, the third case is that the two commands are on comparable nodes and one of them
is an assertion command, when the proposition is again trivially true.
\end{proof}

Lemma\nobreakspace \ref {combine_independent_commands} extends to sequences as well:

\begin{mylem}\mlabel{combine_independent_sequences}
The combination of independent sequences is defined on all filesystems
where both of the sequences are defined:
\[ S\indep T \Longrightarrow \worksmbx{S,T}{S\cc T}. \]
\end{mylem}
\begin{proof}
Assume, by contradiction, that there is a filesystem $\FS$ so that
$S\aFS\neq\fsbroken$ and $T\aFS\neq\fsbroken$, but
$(S\cc T)\FS=\fsbroken$.

From Definition\nobreakspace \ref {def_indep}
the commands in $S$ and $T$ pairwise commute, and so any sequence
that contains the commands from $S$ and $T$ and preserve their original partial order
is equivalent to $S\cc T$ on all filesystems.

Let the command in $T$ that breaks $\FS$ when applying $S\cc T$ be $t$
so that $T=T_0\cc t\cc T_1$.
It is still true that $(T_0 \cc t)\FS\neq\fsbroken$,
and by definition $(S\cc T_0)\FS\neq\fsbroken$,
but $(S\cc T_0\cc t)\FS=\fsbroken$.
Also, from above we know that $S\cc T_0\equiv T_0\cc S$
and so $(T_0 \cc S)\FS\neq\fsbroken$.

If we denote the first command in $S$ with $s_1$,
this means that $(T_0 \cc s_1)\FS\neq\fsbroken$,
which we can combine with $(T_0 \cc t)\FS\neq\fsbroken$, $t\indep s_1$ and
Lemma\nobreakspace \ref {combine_independent_commands}
(using $T_0\FS$ as the reference filesystem)
to arrive at $(T_0 \cc s_1\cc t)\FS\neq\fsbroken$.

We can repeat this step for $s_2$, the next command in $S$,
and from 
$(T_0 \cc s_1\cc t)\FS\neq\fsbroken$
and
$(T_0 \cc s_1\cc s_2)\FS\neq\fsbroken$
arrive at
$(T_0 \cc s_1\cc s_2\cc t)\FS\neq\fsbroken$.
This can be repeated until $S$ is exhausted and we get
$(T_0 \cc S\cc t)\FS\neq\fsbroken$, which is a contradiction.
\end{proof}

We also prove the following:

\begin{mylem}\mlabel{worksinputmatch}
If $S$ and $T$ are minimal sequences, $\worksmnilb{S,T}$,
and there are commands $\cxynv\in S$ and $\czwnv\in T$ on the same node $n$,
then the input types of these commands must match, i.e. $[x]=[z]$.
\end{mylem}
\begin{proof}
This result is similar to Lemma\nobreakspace \ref {equiv_simple_same_commands}.
If, by contradiction, $[x]\neq [z]$, then there would be no filesystem that
either $\cxynv$ or $\czwnv$ would not break, 
and consequently $S$ and $T$ could not work on the same filesystem.
\end{proof}

\subsection{The Correctness of Reconciliation}

We are now ready to prove that the proposed algorithm for reconciliation is correct,
that is, applying its result is not going to break the replicas.
We reformulate the original proposition ($\reca\aFS_B\neq\fsbroken$)
based on $\FS_B=B\aFS$, and so we aim to prove that
$B\cc\reca$ is defined wherever $A$ and $B$ are defined:

\begin{myth}\mlabel{reconciliation_correct}
If $A$ and $B$ are simple, then $\worksmbx{A,B}{B\cc \reca}$,
where, to restate Definition\nobreakspace \ref {def:reconciliation},
\[ \reca = \ords{\alpha \whr \alpha\in \ambnp  \mbox{~and~}  \alpha\indep \bmanp }. \]
\end{myth}

This is trivial unless $\worksmnilb{A,B}$, so we assume that it is true.
Let us first investigate the part of $A$ and $B$ that is excluded from
$\reca$: their intersection.

\begin{mylem}\mlabel{can_move_intersection}
Let $A$ and $B$ be two simple sequences for which $\worksmnilb{A, B}$.
Then $A$ can be separated into their intersection and the remaining commands:
\[ \acb\cc\amb \in \orderset{A}. \]
\end{mylem}

\begin{proof}
Mark those commands in $A$ that are also in $B$.
We show that
there is an $A'$ in $\orderset{A}$ in which all marked commands are at the beginning.

We know that if the marked commands are not at the beginning, then
the sequence contains an unmarked command followed by a marked command.
We show that these can be swapped resulting in an equivalent sequence.
Then, by repeating this process similarly to bubble sorting, we arrive at 
a suitable permutation of $A$, which is also equivalent to $A$, and therefore is in $\orderset{A}$.

Let us consider therefore the marked command preceded by an unmarked command in $A$,
and let the marked command be $\cxynv$, 
and the preceding unmarked command be $\czwmv$:
\begin{gather*}
A = \cdots\cc  \czwmv\cc  \cxynv\cc  \cdots
\end{gather*}
As $A$ is simple and $\wrks{A}$, from 
Rule\nobreakspace \ref {ax_distantrel_breaks} and\nobreakspace Lemma\nobreakspace \ref {simple_distant_pairs}
we know that these commands can only be on incomparable nodes or form a construction or destruction pair.
In the first case, swapping the commands results in a sequence equivalent to $A$,
and we show that the last two cases are impossible.

In the cases then $\czwmv\cc\cxynv$ is either a construction or a destruction pair,
and, depending on which is the case,
$m=\parent(n)$ or $n=\parent(m)$, respectively.

If $B$ has a command on $m$, let it be $\cqrmv$.
As Lemma\nobreakspace \ref {simple_distant_pairs} applies to $B$, we know that
$\cxynv$ (which is in $B$), and $\cqrmv$ must also form a construction
or destruction pair.
As $\cxynv$ cannot be both a construction and a destruction command,
this pair must be of the same type as $\czwmv\cc\cxynv$,
and as the relationship between $n$ and $m$ is given,
we know $\cqrmv$ must precede $\cxynv$ in $B$.

Because the output value of the first command in construction
and destruction pairs is determined by the type of the pair,
$r=w$.
Also, as $\worksmnilb{A,B}$, 
from Lemma\nobreakspace \ref {worksinputmatch}
we know that $[q]=[z]$.
Therefore $\cqrmv=\czwmv$,
which is a contradiction as $\czwmv$ was not marked,
but we see it must also be in $B$.

The last remaining case is when there is no command on $m$ in $B$.
Let $B'$ be $B$
with an extra assertion command added just before $\cxynv$
according to Rules\nobreakspace \ref {ax_child_assert} and\nobreakspace  \ref {ax_parent_assert}, 
from which we know that $B'\equiv B$.
Let the new command be $\cqrmv$.
As $B'$ is still minimal (but no longer simple),
the argument above applies, and we again know
that $[q]=[z]$ and $r=w$.
As $\cqrmv$ is an assertion command, $[q]=[r]$ also holds,
from which $[z]=[w]$ (as $[z]=[q]=[r]=[w]$), which is a contradiction
as $A$ contains no assertion commands.
\end{proof}

We aim to prove that $\worksmbx{A,B}{B\cc \reca}$.
From Lemma\nobreakspace \ref {can_move_intersection} above we know that we can move commands in $\acbnp$
to the beginning of $A$ and $B$, and so this claim is equivalent to
\[ \worksm{\big\{\acb\cc\amb,~ \acb\cc\bma\big\}}{\acb\cc \bma\cc \reca}. \]

For ease of reading, let us rename
 $\amb$ to $S$,
 $\reca$ to $S^*$,
 $\bma$ to $T$,
 and $\acb$ to $U$.
In this notation, we intend to prove that
\[ \worksmbx{U\cc S,U\cc T}{U\cc T\cc S^*}. \]
We can do so if we can show that
\[ \worksmbx{S,T}{S^*}. \]

This is because from $\worksmbx{S,T}{S^*}$ trivially $\worksmbb{S,T}{S^*, T}$,
and from $S^* \indep T$ and Lemma\nobreakspace \ref {combine_independent_sequences} we know 
$\worksmbx{S^*,T}{T\cc S^*}$.
Combining the two we get $\worksmbx{S,T}{T\cc S^*}$,
which using Claim\nobreakspace \ref {works_restricted} yields
$\worksm{U\cc\{S,T\}}{U\cc T\cc S^*}$, that is,
$\worksmbx{U\cc S,U\cc T}{U\cc T\cc S^*}$.

To restate some results above, we therefore already know the following:
\begin{itemize}
\item $S=\amb$ and $T=\bma$ are simple as $A$ and $B$ are simple;
\item $S\cap T = \amb \cap \bma = \emptyset$;
\item $S^*$ is the largest subset of $S$ for which $S^*\indep T$, as
$\reca$ is the largest subset of $\amb$ for which $\reca\indep\bma$.
\end{itemize}
Our theorem is therefore proven if we prove that these imply the proposition above:
\newcommand{\condSimple}{(c1)}
\newcommand{\condDisj}{(c2)}
\newcommand{\condApr}{(c3)}
\begin{mylem}\mlabel{reconciliation_correct_part}
If
{\rm\condSimple} $S$ and $T$ are simple sequences,
{\rm\condDisj} $S\cap T=\emptyset$,
and {\rm\condApr} $S^*$ is the largest subset of $S$ where $S^*\indep T$,
then
$\worksmbx{S,T}{S^*}$.
\end{mylem}
\begin{proof}
This is trivial unless $\worksmnilb{S,T}$, so we assume that it is the case.
The proof is similar to that of Lemma\nobreakspace \ref {can_move_intersection}.
We mark all commands in $S$ that are in its subset $S^*$, and
we prove that $S$ can be transformed into $S'$
where all marked commands are at the beginning.
If so, then 
$\worksm{S}{S^*}$ based on Claim\nobreakspace \ref {worksextpostfix},
from which we get $\worksmbx{S,T}{S^*}$.

Again we know that if $S$ does not already have all marked commands at its beginning,
then there is an unmarked command followed by a marked one.
We show that these commands are independent and so they can be swapped
resulting in an equivalent sequence.
By repeating this process we can generate a suitable $S'$.

Let therefore the marked command in $S$ be $\cxynv$
and the preceding unmarked command be $\czwmv$.
As $S$ is simple and $\wrks{S}$, from 
Rule\nobreakspace \ref {ax_distantrel_breaks} and\nobreakspace Lemma\nobreakspace \ref {simple_distant_pairs}
we know that these commands can only be on incomparable nodes or form a construction or destruction pair.
In the first case, swapping the commands results in a sequence equivalent to $S$,
and we show that the other two cases are not possible as they would lead to contradiction.

In the last two cases, we know that either $m=\parent(n)$ or $n=\parent(m)$.
We also know that because of {\condApr} there must be 
a command $\cqrov$ in $T$ which is not independent of $\czwmv$
as $\czwmv$ is not part of $S^*$.
Based on Corollary\nobreakspace \ref {incomparable_is_independent} we know this means that
either $m\descendantEq o$ or $o\descendantEq m$.
We know none of these commands is an assertion command, and 
from {\condDisj} that $\cqrov\neq\czwmv$.
Also, from {\condApr} we know that $\cxynv\indep \cqrov$,
and so because of Corollary\nobreakspace \ref {incomparable_is_independent},
$n\unrel o$.

We therefore have four cases considering the relationships between $n,m$ and $o$:
\begin{itemize}
\item $n=\parent(m)$ and $o\descendantEq m$.
   This would mean that $o\descendantEq n$ or $n=\parent(o)$ (if $o=m$), which contradicts $n\unrel o$.
\item $n=\parent(m)$ and $m\descendantEq o$.
   This would mean that $n\descendant o$, which contradicts $n\unrel o$.
\item $m=\parent(n)$ and $o\descendantEq m$.
   This would mean that $o\descendant n$, which contradicts $n\unrel o$.
\item $m=\parent(n)$ and $m\descendantEq o$.
   Let us continue with this case.
\end{itemize}

As $m=\parent(n)$, we know $\czwmv\cc\cxynv$ must be a construction pair.
We know that $T$ cannot have a command on $m$ as $m=\parent(n)$ and $\cxynv\indep T$,
and so $m\neq o$ and therefore $m\descendant o$.
This means we can create a new sequence $T'\equiv T$ by inserting a command on 
$m$ into $T$ before $\cqrov$
according to Rule\nobreakspace \ref {ax_parent_assert}, the input type of which is $[\vald]$.
As $T'$ is still minimal, and $\worksmnilb{S,T}$,
from Lemma\nobreakspace \ref {worksinputmatch} we get $[z]=[\vald]$, which is a contradiction,
as $\czwmv$ is the first command in a construction pair.
\end{proof}

\subsection{Reconciliation is Maximal}

The reconciliation algorithm defined above is also maximal, that is,
it is not possible to apply any further commands from $\ambnp$ to $\FS_B$.
To show this, we are going to prove that any sequence $S$ formed from commands
in $\ambnp$ that is not independent of $\bmanp$ necessarily breaks $B\aFS=\FS_B$
or introduces a conflicting update.

\begin{myth}\mlabel{rec_is_complete}
If $A$ and $B$ are simple sequences,
and $S$ is a subsequence of $S'\in\orderset{\ambnp}$ that
contains a command $\cxynv$ for which $\cxynv\nindep\bmanp$,
then $B\cc S$ either breaks all filesystems $A$ and $B$ are defined on
(that is, $S$ breaks all possible $\FS_B$ replicas),
or $S$ changes a node that $B$ has already changed to a different value,
that is, it overrides a change in $\FS_B$.
\end{myth}
Such an override could occur if a given node was modified differently in
$A$ and $B$ (to different values but to values of the same type), which 
our algorithm must treat as a conflict to be resolved by
by the user or a different system.
\begin{proof}
The proof is similar to the ones we have seen above.
Let $T$ be an arbitrary sequence from $\orderset{\bmanp}$.
Without loss of generality, we assume that $\cxynv$ is the first command in $S$
that is not independent of $T$.
If so, we can split $S$ into $S_0\cc\cxynv\cc S_1$ where $S_0\indep T$.

Let $\czwmv$ be the last command in $T$ that is not independent of $\cxynv$,
and split $T$ into $T_0\cc\czwmv\cc T_1$, where therefore $T_1\indep\cxynv$.

From Lemma\nobreakspace \ref {can_move_intersection} we know that
$\acb\cc\bma \in \orderset{B}$, and so
$B\cc S \equiv \acb\cc\bma\cc S \equiv \acb\cc T\cc S$.
We also know that commands in $S_0$ and $T$ commute, and so this is equivalent to
$\acb\cc S_0\cc T\cc\cxynv\cc S_1$.
Expanding $T$ and swapping $T_1$ and $\cxynv$ we get
\[ \acb\cc S_0\cc T_0 \cc \czwmv \cc \cxynv\cc T_1\cc S_1. \]

First, we prove that this sequence breaks all filesystems
where $A$ and $B$ are defined unless $n=m$.
Let us therefore suppose $n\neq m$ and, to use an inverse proof,
that the sequence is defined on $\FS$ 
where $A\aFS\neq\fsbroken$ and $B\aFS\neq\fsbroken$.
If so, its initial segment,
\[ \acb\cc S_0\cc T_0 \cc \czwmv \cc \cxynv, \]
must also be defined on $\FS$.

We know $\cxynv\nindep\czwmv$, and so from Corollary\nobreakspace \ref {incomparable_is_independent}, $n\nunrel m$.
As the sequence is defined on $\FS$,
because of Lemma\nobreakspace \ref {simple_distant_pairs}
$\czwmv\cc\cxynv$ can only be a construction or destruction pair.

Also, because $B$ is simple and contains $\czwmv$,
we know there are no commands on $m$ 
in either $\acb$ or $T_0$, which are both formed from other commands in $B$.
Moreover, we know there are no commands on $m$ in $S_0$ because $S_0\indep T$.
We therefore know that $\FS(m)$ must be of type $[z]$.

If $m=\parentf{n}$ then $\czwmv\cc\cxynv$ is a construction pair, and we know
$[\FS(m)]=[z]\neq[\vald]$ and $y\neq\empt$ as none of these commands is an assertion command.
If so, $\fsreplacement{\FS}{n}{y}$ 
(that is, applying $\cxynv$ to $\FS$ without modifying $\FS(m)$ first) would violate the tree property.
The same is true when $\parentf{m}=n$, where we have a destruction pair,
and $[z]\neq\empt$ and $[y]\neq[\vald]$.

Since $\cxynv\in S$ and therefore $\cxynv\in A$, and $A\aFS\neq\fsbroken$,
the above means that $A$ must also contain a command on $m$ before $\cxynv$.
From Lemma\nobreakspace \ref {simple_distant_pairs} we then know that $\cqrmv\cc\cxynv$
is also a construction or destruction pair, and it is the same type as
$\czwmv\cc\cxynv$.

Because the type of the pair determines the output type of the first command,
this means that $[r]=[w]$, and they can either be $[\empt]$ or $[\vald]$.
As we assumed that $|\setd|=1$, we therefore know that $r=w$.
From Lemma\nobreakspace \ref {worksinputmatch} we also know that $[q]=[z]$.

Therefore $\cqrmv = \czwmv$, which is a contradiction
because $\cqrmv$ was selected from $A$, and
$\czwmv$ was selected from $T$, that is, from $\bmanp$.

The only possibility is therefore $n=m$.
As $A$ and $B$ are simple, we then know that
$\cxynv$ and $\czwmv$ are the first commands on $n$,
and so (again from Lemma\nobreakspace \ref {worksinputmatch}) $[z]=[x]$.
We also know $w\neq y$ as otherwise the two commands would be equal and
would be in $\acbnp$.
This, however, means that $\cxynv$, from $A$, overrides a change introduced 
by $\czwnv$, from $B$, which must be treated as a conflict.
\end{proof}

\section{Extending the Synchronizer}\mlabel{sec_multidir}

The two assumptions we used to increase the symmetry of our model were
that there is no \emph{move} command,
and that
there is only one directory value, that is, directories are not differentiated
by the meta-information they contain.
In this section we describe reasons why the latter assumption may not limit the applicability
of even the current model, and we also describe encoding and decoding steps
to overcome the two restrictions.

We note that commercial designs for synchronizers 
often avoid considering metadata in directories as
as these are generally not understood well by users,
and, if needed, conflict resolution on these settings can be easily automated.
See, for example, \cite{BZ} and the seminal work of Balasubramaniam and Pierce and the Unison synchronizer \cite{BP}.

Despite the above, there can be applications where directory metadata
is considered important.
Synchronizers based on our model can be readily extended to handle them
by duplicating $\setn$. 
Given the input filesystems over $\setn$, we add a special node as an extra child under each original node,
and we encode directory metadata in a file under each directory.
It is easy to see that update detection and conflict resolution can continue as expected,
with the only exception of a potential conflict detected on one of these special nodes.
In these cases the synchronizer may prescribe creating a directory without creating
the special metadata file, which is clearly not possible to do on the target filesystem,
as creating a directory entails specifying its metadata as well.
In these cases, the synchronizer can either fall back to default values and flag the issue
for review later, or if, as suggested above, resolution of conflicts on the metadata (e.g. readable, writable and executable flags)
can be easily automated, then it could form part of the implementation.

Reintroducing the \emph{move} command can happen in a similar fashion.
Its main advantages include that 
\emph{move} commands are easier to review,
and if the synchronizer suggests or performs a move instead
of deleting a file and recreating it somewhere else, the user
can be assured that no information is lost.
During pre-processing, if an update detector is used,
\emph{move} commands still do not need to be part of the command sequences
the synchronizer operates on.
If the command sequences are derived from filesystem journals, they
can be encoded as separate delete and create operations to allow the existing reconciliation
algorithm to operate.

In a post-processing stage, deleted or overwritten file contents
can be paired up with created ones to reintroduce \emph{move}
commands in the output sequences and take advantage of their benefits.
The renaming of whole filesystem subtrees can also be detected
to further aid the user.

\section{Towards a Filesystem-Free Algebra}\mlabel{sec_algebra}

Similarly to \cite{NREC} we can consider creating an algebra over command sequences
which would enable us to draw conclusions about the behaviors of the sequences
without considering the underlying filesystems.
This would provide a secondary model of (filesystem) commands above (and independent of) the model of filesystems
defined in this paper. The new model would no longer describe filesystems as such,
but would use known relationships between sequences of commands as its starting point.

In this algebra, equivalence ($\equiv$) and extension ($\eqext$) would become algebraic relations
between sequences, and logical rules involving them
(e.g. $ A\equiv B \Rightarrow S\cc A\cc T\equiv S\cc B\cc T $) would be re-cast as inference rules.
In addition to these, the Rules listed in Section\nobreakspace \ref {section_axioms}
would become axioms that we accept as true, and from which other true statements can be derived
using the inference rules.
Such a system would allow us to deduce, for example, whether two sequences are equivalent
or one extends the other without reference to filesystems.

We note, however, that the Rules to be used as axioms 
specify the relationships between the nodes
the commands act on, which would require a more complicated set of symbols to represent
in the new algebra.
To avoid this, we think it will be more fruitful to regard commands (and sequences)
on different nodes as different, and, in effect, have nine times as many commands 
in the algebra as there are nodes in the namespace of the filesystems.
This will allow us to reduce the number of types of symbols in our algebra, and regard
the Rules not as axioms, but as templates for axioms
from which all axioms (for pairs of separate nodes, etc.) can be generated.

We expect that the soundness and completeness of such an algebra can be proven
similarly to the proofs described in \cite{NREC}.
Indeed, results in the current paper can serve as building blocks of such a proof:
derivations in Appendix\nobreakspace \ref {axiom_proof} for the Rules
in effect proves the soundness of all axioms,
and Lemma\nobreakspace \ref {simple_reorder_equiv} proves completeness in a limited sense.
It is also worth noting that the majority of proofs presented in this paper
do not actually refer to specific filesystems or the definition of commands,
but draw on known relationships between the commands themselves.
In other words, many proofs are transferable to the algebra that we describe here.

In a more complete algebra from which the correctness and completeness
of reconciliation could also be derived,
one would also include symbols, inference rules and axioms for 
$\worksmnil{\seqset{A}}$ and $\worksm{\seqset{A}}{\seqset{B}}$.
We see proving the soundness and completeness of this extended algebra
an intriguing problem that is worthy of further research.

\section{Conclusions and Further Research}\mlabel{sec_conclusion}

In this paper we presented an algebraic model to file synchronization,
and defined an update detector and a reconciliation algorithm.
With its extended command set the model was carefully designed to take
advantage of an inherent symmetry of filesystem commands,
which was enhanced by removing meta-information from directories
as well as the \emph{move} command.
These assumptions enabled us to prove that
the update detector algorithm is correct (Theorem\nobreakspace \ref {update_works}),
and that reconciliation is not only correct (Theorem\nobreakspace \ref {reconciliation_correct})
but also yields maximal results (Theorem\nobreakspace \ref {rec_is_complete}).

We introduced relations and concepts that have proven particularly
useful, including 
independent commands, %
simple sequences, %
and relationships between domains of sets of command sequences ($\worksmeqsign$)
that allowed us to formalize the correctness of reconciliation.
Among our partial results were interesting properties of these concepts,
for example,
that all permutations of simple sequences satisfying a partial order are equivalent (Lemma\nobreakspace \ref {simple_reorder_equiv}),
that command sequences can be transformed into simple sequences using syntactical rules (Theorem\nobreakspace \ref {can_simplify}),
or that the combination of independent sequences is defined wherever the components are defined (Lemma\nobreakspace \ref {combine_independent_sequences}).
We believe these results will prove useful outside their current usage as well.

Our results can be immediately applied to practical filesystem synchronizers,
which could benefit from e.g. recording changes using the extended command set,
and from the theoretical background that guarantees their correct operation.
Furthermore, with the encoding and decoding steps described in Section\nobreakspace \ref {sec_multidir},
the reconciliation algorithm can be extended to handle directory meta-information
and the \emph{move} command,
which makes our model applicable to more use cases and filesystems.

\myskip
Apart from constructing a filesystem-free algebra,
there are many other ways in which further work can extend the current results.
An important extension would be to
consider reconciling not only two, but more replicas in a single step and
prove the correctness and maximality of the algorithm proposed,
or show that it is impossible to satisfy these criteria.

A related problem is to extend the system and the proofs
to allow for cases where reconciliation cannot
complete fully
or if only a subset of the replicas are reconciled 
(e.g. due to network partitioning),
both of which would result in a state where different replicas
have different common ancestors, that is,
the updates specific to the replicas start from different points
in the update history of the filesystems.
Existing research can offer pointers as to how such cases can be modeled
in our algebraic system.
Parker et al. \cite{PPRS} and Cox and Josephson \cite{CJ}
describe version vectors (update histories) kept as metadata,
while Chong and Hamadi present distributed algorithms that allow incremental synchronization \cite{CH}.
Representing individual updates to files
in their modification histories (as described in \cite{CJ})
as separate commands could also enable an algebraic synchronizer to reconcile otherwise
conflicting updates and resolve partial reconciliations.

Future work could also investigate extending the 
model and algorithms to the $|\setd|>1$ case
so that directory metadata could be represented directly
as opposed to through encoding and decoding steps.

And finally, we hope that this work, together with \cite{NREC}, provides
a blueprint of constructing an algebra of commands for different storage protocols
(e.g. XML trees, mailbox folders, generic relational databases, etc.),
and of demonstrating the adequacy and completeness of the update and conflict detection and reconciliation
algorithms defined over it.
This, in turn, can offer formal verification of the algorithms underlying
specific implementations in a variety of synchronizers.
Alternatively, by generalizing the parent--child relationships between filesystem nodes,
the demonstrated properties of minimal sequences of commands
and domains of sets of sequences ($\worksmeqsign$)
may also contribute to future research into algebraic structures
constrained by predefined sparse connections between their elements.

\section{Acknowledgments}

I thank L.~Csirmaz for his invaluable input on this paper,
and Bill Zissimopoulos for his observations on algebraic file synchronization that triggered this research.
I am indebted to the anonymous reviewers of the previous version of this paper
for their helpful and detailed remarks
that helped improve the presentation of the paper
and make it accessible for a wider audience.

\appendix %

\section{Related Work}\mlabel{sec_relatedwork}

\subsection{Liberal and Conservative Reconciliation}

We consider the reconciliation algorithm described here an improvement over
the one derived during our previous research \cite{NREC}
as the previous algorithm not only fails to propagate all possible commands
wherever it is possible (that is, where it does not break the filesystem),
but the current algorithm is also simpler.
This is because the previous reconciliation algorithm excludes
commands from being propagated which must be preceded by a command that conflicts.

The former observation is supported by Berzan and Ramsey, who in \cite{CBNR} 
describe different general reconciliation policies.
The liberal (maximal) policy propagates all updates to all replicas where
the update does not break the filesystem, while a conservative policy
refrains from updating any node that is below a node with conflicting commands.
They show that the reconciliation algorithm in \cite{NREC} implements
an intermediate policy as one can easily construct two- or three-replica scenarios
where an update could clearly be propagated, but it is excluded.
To rephrase the example in \cite{CBNR}, consider the following two
update sequences that have been applied to replicas $\FS_A$ and $\FS_B$:
\begin{align*}
A&=\caaa{\empt}{\valdi}{\pn}\cc\caaa{\empt}{\valf}{n} \\
B&=\caaa{\empt}{\valdii}{\pn}
\end{align*}
(\cite{NREC} did not require all directories to have the same value).
Clearly $\caaa{\empt}{\valf}{n}$ could be applied to $\FS_B$, but it is not as
it must be preceded by creating the directory, which conflicts with the same command in $B$
due to the different value.
(Introducing a third replica which has not changed at all further complicates the picture.)
The current algorithm no longer needs to specify that there can be no conflicts
on preceding commands, and, to use the above terminology, implements a fully liberal policy.

\subsection{Comparing Definitions of Conflicts}

As a test of the proposed reconciliation algorithm, 
we compare our definition of conflicting
updates to how conflicts are defined by Balasubramaniam and Pierce
in their state-based approach implemented in the Unison synchronizer \cite{BP}.
As we noted in \cite{NREC}, the update detector they describe provides a safe estimate of nodes
(paths) at which updates occurred.
It marks some nodes as \emph{dirty} in a way that we know that at non-dirty nodes
the filesystems (replicas) have not changed between their common original state and current state
(see Definition 3.1.1 in \cite{BP}).
The \emph{dirty} marks are also up-closed, that is, all ancestor nodes of a dirty node
are also dirty (Fact 3.1.3 in \cite{BP}).
And finally, in our notation, 
a conflict is detected between replicas $\FS_A$ and $\FS_B$ at node $n$
if $n$ is marked as dirty in both $\FS_A$ and $\FS_B$, and
$\FS_A(n)\neq\FS_B(n)$, and $n$ does not point to a directory in both replicas.
(In other words, 
$[\FS_A(n)]\neq[\vald]$ or $[\FS_B(n)]\neq[\vald]$; see section 4.1 in \cite{BP}.
Let us note that there is an alternative approach to defining an algorithm for
the same synchronizer by Pierce and Vouillon in \cite{PV}.)

It can be easily seen that due to an edge case, 
not all conflicts detected based on the above definition
entails a conflict based on our system; that is, \cite{BP} describes a more
conservative policy.
We use the example we described in \cite{NREC},
where $\FS(\pn)$ is a directory and $\FS(n)$ is a file, and
the two replicas are derived in the following way:
\begin{gather*}
\FS_A = (\cfba{n} \cc \cdba{\pn}) \aFS \\
\FS_B = \cfba{n} \aFS.
\end{gather*}
Then, $\pn$ is dirty in both replicas:
in $\FS_A$ it was modified, and in $\FS_B$ one of its descendants was modified.
Moreover, $\FS_A(\pn)\neq\FS_B(\pn)$ and $[\FS_A(\pn)]\neq[\vald]$ as it is empty.
Therefore, a conflict is detected at $\pn$.
(This behavior is preserved in the more recent Harmony synchronizer.
See ``delete/delete conflicts'' in \cite{PSG,FGKPS}.)
Our reconciliation algorithm detects no conflicts;
instead, it propagates $\cdba{\pn}$ to $\FS_B$, which we think is as expected
and desired.

At the same time, it can be shown that if our command-based reconciler
detects a conflict, it entails a conflict in the state-based reconciler.
We note here that Balasubramaniam and Pierce also suppose that all directories are equal,
therefore, as elsewhere, we are safe to continue to assume that $|\setd|=1$.

\begin{proof}
Let $A$ and $B$ be two simple sequences returned by the update detector
for the two replicas $\FS_A$ and $\FS_B$.
A conflict between the commands $\cxynv\in A$ and $\czwmv\in B$
means that even though $\cxynv\in A\setminus B$
and $\czwmv\in B\setminus A$, they cannot be included in
$\reca$ and $\recb$, respectively, because
$\cxynv\nindep\czwmv$ (see Definition\nobreakspace \ref {def:reconciliation}).

From Corollary\nobreakspace \ref {incomparable_is_independent} we therefore know that $n\nunrel m$.
Without loss of generality we can assume that $n\descendantEq m$.
From this we see that $n$ is \emph{dirty} in both $\FS_A$ and $\FS_B$,
as the filesystem changes at $n$ in $\FS_A$ and at $n$ or its descendant in $\FS_B$
(and none of the commands are assertion commands).

Now we only need to show that $\FS_A(n)\neq\FS_B(n)$, because from this 
and $|\setd|=1$ we will also
know that at least one of these values is not a directory,
and therefore there is a conflict at $n$ in the state-based reconciler.

If $n=m$, this follows from the fact that $\cxynv\neq\czwmv$ 
(because they are not in $A\cap B$),
as from Lemma\nobreakspace \ref {worksinputmatch} we have $[x]=[z]$,
and therefore necessarily $y\neq w$,
that is, $\FS_A(n)=y\neq w=\FS_B(n)$.

If $n\descendant m$ and there is no command on $n$ in $B$, then 
because $\cxynv$ is not an assertion command, we again have
$\FS_A(n)\neq\FS(n)=\FS_B(n)$, where $\FS$ is the common ancestor of the replicas.
Finally, if there is a command on $n$ in $B$, then 
from $\cxynv\in A\setminus B$ we know
it must be different from $\cxynv$, and similarly to the first case,
we have $\FS_A(n)\neq\FS_B(n)$.
\end{proof}

\section{Proofs for Rules}\mlabel{axiom_proof}

In this section we derive the Rules listed in
Lemma\nobreakspace \ref {rules_lemma} from our filesystem model.
In the proofs, we may say that a filesystem
is \emph{broken at $n$}
if after applying a command, it violates the tree property
because
$n$ is not a dictionary, but has at least one non-empty child.

It follows from the definition of the tree property that
whether a filesystem is \emph{broken at a node $n$} is determined only
by the values at $n$ and the children of $n$.

Below we restate the Rules defined in
Lemma\nobreakspace \ref {rules_lemma}.

\begin{myaxproof}
\axaxseparatecommute
\end{myaxproof}

\begin{proof}
This proof is similar to the first case in the proof of Lemma\nobreakspace \ref {combine_independent_commands}.
However, this result is needed for Lemma\nobreakspace \ref {independent_details} on which that proof rests.

We use an inverse proof and assume $\cxynv\cc\czwmv \nequiv \czwmv\cc\cxynv$.
Since the values in the resulting filesystems are the same,
this is only possible if, for an initial non-broken filesystem $\FS$,
one side results in a broken filesystem, and the other side does not.
Without loss of generality, we can assume
$(\cxynv\cc\czwmv)\aFS$ is not broken, but $(\czwmv\cc\cxynv)\aFS$ is.
This means that either $\czwmv\aFS$ is already broken, or it is not, and applying $\cxynv$
breaks the filesystem.

In the former case 
$\czwmv\aFS$ must be broken at $m$ or $\parentf{m}$, since it only changed at $m$.
Let us look at the first case.
Since $n\indep m$, we know that the parent and children of $m$ have the same value
in $\FS$ as they do in $\cxynv\aFS$. However, this is a contradiction,
as applying $\czwmv$ to $\FS$ leads to a broken filesystem,
but applying it to $\cxynv\aFS$ (from the left side) does not.

Therefore $\czwmv\aFS$ must be broken at $\parentf{m}$.
By definition this means that its value at $\parentf{m}$ cannot be a directory.
Also, reasoning similar to the above shows that this is only possible if $\cxynv$ changes
the environment of $\parentf{m}$, which, since $n\indep m$, is only possible
if $n$ and $m$ are siblings.

Since then the value at $\parentf{n}=\parentf{m}$ is not changed by either command,
we know it cannot be a directory in $\FS$, either, and that therefore
$\FS$ is empty at all children of the parent node.
Since by assumption $\czwmv\aFS$ is broken, $[w]\neq[\empt]$ must hold.
However, this is a contradiction, as this would necessarily mean that
the left side, $\czwmv(\cxynv\aFS)$ must also be broken.

\begin{sloppypar}
We can proceed in the same fashion if $\czwmv\aFS$ is not broken, but
$(\czwmv\cc\cxynv)\aFS$ is.
\end{sloppypar}
\end{proof}

\begin{myaxproof}
\axaxseparatenobreaks
\end{myaxproof}

\begin{proof}
It is easy to see that over any $\setn$ one can construct a filesystem that neither command breaks.
Since $n\indep m$, it is always possible to set all descendants of both $n$ and $m$ empty,
and have a directory at all ancestors of both $n$ and $m$. In such positions any values are permissible,
so neither command will break the filesystem.
\end{proof}

\begin{myaxproof}
\axaxsamebreaks
\end{myaxproof}

\begin{proof}
This is trivial, as either $\cxynv$ breaks the filesystem, or, if it does not, then
we know $[\cxynv\aFS(n)]=[y]\neq [z]$, and therefore $\czwnv$ will break the filesystem.
\end{proof}

\begin{myaxproof}
\axaxsameemptyseq
\end{myaxproof}

\begin{proof}
From the conditions we see that for every $\FS$,
$\FS$ and $(\cxynv\cc\czwnv)\aFS$ have the same values at every node,
while the latter can still be broken if $\cxynv\aFS$ is broken.
This is equivalent to saying that where the $\cxynv\cc\czwnv$ function is defined,
it is the identity function, that is, it is extended by $\emptyseq$.
\end{proof}

\begin{myaxproof}
\axaxsamesinglec
\end{myaxproof}

\begin{proof}
It is easy to see that where both $\cxynv\cc \czwnv$ and $\cxwnv$ are defined
(they do not break the filesystem), they are equivalent.
What remains to show is that they break the same set of filesystems.
If $(\cxynv\cc \czwnv)\aFS$ is not broken, then we know that neither of
$\FS$, $\fsreplacement{\FS}{n}{y}$ or $\fsreplacement{\FS}{n}{w}$ is broken,
and therefore $\cxwnv$ does not break $\FS$.
Conversely, if $\cxwnv$ does not break $\FS$, then 
we know that $\FS$ and $\fsreplacement{\FS}{n}{w}$ are not broken.
Also, as
either $[x]\neq[\empt]$ or $[w]\neq[\empt]$,
there must be a directory at $\FS(\parentf{n})$,
and so $\fsreplacement{\FS}{n}{y}$ cannot be broken, either,
as any value is permissible at $n$.
This is equivalent to saying that $\cxynv\cc \czwnv$ does not break the filesystem.
We note that if $\cxwnv$ is an assertion command, then it extends $\cxynv\cc\czwnv$
instead of being equivalent to it.
\end{proof}

\begin{myaxproof}
\axaxdirectchildbreaks
\end{myaxproof}

\begin{proof}
If neither of the commands is an assertion command, then the claim follows from
Lemma\nobreakspace \ref {simple_distant_pairs} in the following way.
We use an inverse proof and assume that $\cxypnv\cc\czwnv$ is not a construction pair,
none of the commands is an assertion command, and $\wrksx{\cxypnv\cc\czwnv}$.
Therefore Lemma\nobreakspace \ref {simple_distant_pairs} applies to this sequence, and the two commands
must form a construction or destruction pair. This is a contradiction as due to
the relationship between the nodes, they cannot form a destruction pair, and
they do not form a construction pair by assumption.

Otherwise, let $\FS_1$ be $\cxypnv\aFS$, and $\FS_2$ be $(\cxypnv\cc\czwnv)\aFS$.
If $\cxypnv=\cbba{\pn}$ (the only assertion command it can be given the conditions),
then either $\FS_1$ is broken, or
$\FS(\pn)=\FS_1(\pn)=\FS_2(\pn)=\empt$.
Since $\czwnv\neq\cbba{n}$, we know that either $\czwnv$ breaks $\FS_1$,
or, depending on $z$ and $w$, either $\FS_1(n)\neq\empt$ or $\FS_2(n)\neq\empt$.
In any case $\FS_2$ will be broken at $\pn$ which is empty but at one point has a non-empty child.

Finally, if $\czwnv=\cdda{n}$, we can proceed in a similar way,
as since $\cxypnv\neq\cdda{n}$, we know at one point the value at $\pn$ is not a directory.
\end{proof}

\begin{myaxproof}
\axaxdirectparentbreaks
\end{myaxproof}

\begin{proof}
This proof follows the same logic as that of Rule\nobreakspace \ref {ax_directchild_breaks},
using destruction pairs instead of construction pairs.
\end{proof}

\begin{myaxproof}
\axaxdistantrelbreaks
\end{myaxproof}

\begin{proof}
We use an inverse proof and
suppose $\cxynnv$ and $\czwnv$ satisfy the conditions,
and there is an $\FS$ for which $(\cxynnv\cc\czwnv)\aFS\neq\fsbroken$.
If $\czwnv\neq\cbba{n}$, then either $\FS(n)\neq\empt$ or $\czwnv\aFS(n)\neq\empt$,
and therefore $\FS(\parentf{n})$ must be a directory for the filesystem to satisfy
the tree property at all times.
However, from $\cxynnv\neq\cdda{\nn}$ we also know that either $\FS(\nn)$ 
(before applying $\cxynv$) or $\cxynv\aFS(\nn)$ (after applying $\cxynv$)
is not a directory. Depending on which is true, we have a point in the sequence
of applying the commands where the filesystem cannot satisfy the tree property
and therefore becomes broken as $\parentf{n}$, a descendant of $\nn$, contains
a directory and so must all its ancestors.
\end{proof}

\begin{myaxproof}
\axaxchildassert
\end{myaxproof}

\begin{proof}
Since $\cxynnv\neq\cdda{\nn}$, we know that
either $\cxynnv\aFS$ is broken, or, depending on $x$ and $y$,
either $\FS(\nn)\neq\vald$, or $\cxynnv\aFS(\nn)\neq\vald$.
In the latter two cases $\FS$ must be empty at all descendants of $\nn$
for it to satisfy the tree property.
The assertion command $\cbba{n}$ will therefore not break $\FS$.
\end{proof}

\begin{myaxproof}
\axaxparentassert
\end{myaxproof}

\begin{proof}
This proof follows the same logic as that of Rule\nobreakspace \ref {ax_child_assert}.
\end{proof}

\begin{myaxproof}
\axaxassert
\end{myaxproof}

\begin{proof}
The two assertion commands $\cdda{n}$ and $\cbba{n}$ either
break a filesystem, or leave it in the same state.
\end{proof}

\end{document}